\documentclass[runningheads]{llncs}
\usepackage{tikz}
\usepackage{tikz-extension}      
\usepackage{tikz-circlesnarrows1}      

\usepackage{fontawesome}
\usepackage{graphicx}
\usepackage{hyperref}
\usepackage{abbr-habel}
\newcommand{\cqd}{\hspace{0.2in} \vrule  width  6pt  height  6pt  depth 0pt  \vspace{0.1in}}
\usepackage{fleqn,stmaryrd,amssymb,makeidx,theorem,latexsym,amscd,amsmath}
\usepackage[arrow,matrix,tips]{xy}
\usepackage[arrow,matrix,tips]{xypic}
\usepackage{xcolor}

\newcommand{\ignore}[1]{}
\newtheorem{assumption}{Assumption}
\newcommand {\SC} {\mathcal{S}}

\usepackage{a4wide}
\usepackage{float}

\begin{document}
\title{Initial Conflicts for Transformation Rules \\ with Nested Application Conditions}
%
\titlerunning{Initial Conflicts for Transformation Rules with ACs}
%

\author{Leen Lambers\inst{1}\,\textsuperscript{\faEnvelopeO}\orcidID{0000-0001-6937-5167} \and
Fernando Orejas\inst{2}\orcidID{0000-0002-3023-4006}\thanks{F. Orejas has been supported by the Salvador de Madariaga grant PRX18/00308 and by funds from the Spanish Research Agency (AEI) and the European Union (FEDER funds) under grant GRAMM (ref. TIN2017-86727-C2-1-R)}}
\authorrunning{L. Lambers and F. Orejas}
%
\institute{Hasso Plattner Institute, University of Potsdam, Potsdam, Germany
\email{leen.lambers@hpi.de} \and
Universitat Polit\`ecnica de Catalunya, Barcelona, Spain
\email{orejas@lsi.upc.edu}}
\maketitle              
\begin{abstract}
Recently, initial conflicts were introduced in the framework of $\mathcal{M}$-adhesive categories as an improvement to critical pairs, a specific complete subset of all possible conflicts for a pair of transformation rules. Critical pairs are specific in the sense that they show conflicts in a minimal context. They are moreover complete in the sense that for each conflict there exists a critical pair that can be embedded into it (or that represents the same conflict in some minimal context). Initial conflicts represent a proper, but still complete subset of critical pairs. They moreover have the important characteristic that for each conflict a unique initial conflict exists representing it. The theory of critical pairs has been extended in the framework of $\mathcal{M}$-adhesive categories to rules with nested application conditions (ACs), restricting the applicability of a rule and generalizing the well-known negative application conditions. A notion of initial conflicts for rules with ACs does not exist yet. 

We extend the theory of initial conflicts in the framework of $\mathcal{M}$-adhesive categories to transformation rules with ACs. We first show that for rules with ACs, conflicts are in general neither inherited from a bigger context any more, nor is it possible to find a finite and complete subset of finite conflicts as illustrated for the category of graphs. We define initial conflicts to be special so-called symbolic transformation pairs, and show that they are minimally complete (and in the case of graphs also finite) in this symbolic way. We show that initial conflicts represent a proper subset of critical pairs again. We moreover demonstrate that (analogous to the case of rules without ACs) for each conflict a unique initial conflict exists representing it. We conclude with presenting a sufficient condition illustrating important special cases for rules with ACs, where we do not only have initial conflicts being complete in a symbolic way, but also find complete (and in the case of graphs also finite) subsets of conflicts in the classical sense.  
\end{abstract}

\section{Introduction}
\emph{Detecting and analyzing conflicts} is an important issue in software analysis and design, which has been addressed successfully using powerful techniques from graph transformation (see, e.g., \cite{HausmannHT02,KochMP05,Lambers2009,Lambers0TBH18}), most of them based on critical pair analysis. The \emph{power of critical pairs} is a consequence of the fact that: a) they are complete, in the sense that they represent all conflicts; b) there is a finite number of them; and c) they can be computed statically. The main problem is that their computation has exponential complexity in the size of the preconditions of the rules. For this reason, some significantly smaller subsets of critical pairs that are still complete have been defined \cite{LambersEO08,LambersBO0T18,AzziCR19}, clearing the way for a more efficient computation. In particular, recently, in \cite{LambersBO0T18}, a new approach for conflict detection was introduced based on a different intuition. Instead of considering conflicts in a minimal context, as for critical pairs, we used the notion of initiality to characterize a complete set of minimal conflicts, showing that \emph{initial conflicts} form a strict subset of critical pairs. In particular, we have that every conflict is represented by a  unique initial conflict, as opposed to the fact that each conflict may be represented by many critical pairs.

Most of the work on critical pairs only applies to \emph{plain} graph transformation systems, i.e. transformation systems with unconditional rules. Nevertheless, in practice, we often need to limit the application of rules, defining some kind of \emph{application conditions} (ACs). In this sense, in \cite{LambersEO06,Lambers2009} we defined critical pairs for rules with  negative application conditions (NACs), and in \cite{EhrigHLOG10,EhrigGHLO12} for the general case of  ACs, where conditions are as expressive as arbitrary first-order formulas on graphs. However, to our knowledge, no work has addressed up to now the problem of finding significantly smaller subsets of critical pairs for this kind of rules. In this paper we generalize the theory of initial conflicts to rules with ACs in the framework of $\M$-adhesive transformation systems. In particular, the main contributions of this paper (as summarized in \autoref{tab:overview}) are:

\begin{itemize}
\item The definition of the \emph{notion of initial conflict for rules with ACs}, based on a notion of \emph{symbolic transformation pair}, showing that the set of initial conflicts is a \emph{proper subset} of the set of critical pairs and that it is \emph{minimally complete}\footnote{Provided that the considered category has initial conflicts for the plain case.}, in the sense that, no smaller set of symbolic transformation pairs exists that is also complete. In particular, the \emph{cardinality} of the set of initial conflicts is, at most, the cardinality of the set of initial conflicts for the plain case, when disregarding the ACs, plus one. Moreover, as in the plain case, every conflict is an instance of a \emph{unique} initial conflict. 

\item The identification of a class of so-called \emph{regular} initial conflicts that demonstrate a certain kind of regularity in their application conditions.  This allows us to unfold them into a \emph{complete (and in the case of graphs also finite) subset of conflicts}. In particular, we show that, in the case of rules with NACs, initial conflicts are regular, implying that our initial conflicts represent a \emph{conservative extension} of the critical pair theory for rules with NACs.


\end{itemize}

\begin{table*}[htbp]
\small
	\centering
		\begin{tabular}{|p{3.3 cm} | p{3.1cm} | p{4.3cm} | p{3.5cm}|}
		\hline 
    & \textbf{plain rules} & \textbf{rules with NACs} & \textbf{rules with ACs}\\
\hline 
\textbf{critical pairs (CPs)}   & subset of conflicts, \newline complete \newline \cite{Plump93,Plump94,Plump05,EhrigEPT06} & subset of conflicts, \newline complete \newline \cite{LambersEO06,Lambers2009}& symbolic, \newline complete \newline \cite{EhrigHLOG10,EhrigGHLO12} \\
\hline 
\textbf{initial conflicts} & subset of conflicts, \newline min. complete, \newline proper subset of CPs \newline  \cite{LambersBO0T18,AzziCR19} & symbolic (Def.~\ref{def:initial-conflict}), \newline min. complete (Cor.~\ref{cor:minimal-AC}), \newline regular (Thm.~\ref{thm:monotone-NACs}) \& conservative extension of CPs (Thm.~\ref{thm:conservative-NACs})& symbolic (Def.~\ref{def:initial-conflict}), \newline min. complete (Cor.~\ref{cor:minimal-AC}), \newline proper subset of CPs \newline (Thm.~\ref{thm:ic-is-cp})\\
\hline 
		\end{tabular}
		\vspace{.2cm}
	\caption{Critical pairs versus initial conflicts}
	\label{tab:overview}
\end{table*}

%

The paper is organized as follows. We describe \emph{related work} in \autoref{sec:related-work} and, in \autoref{sec:preliminaries}, we present some \emph{preliminary material}, where we also include some new results. More precisely, in \autoref{subsec:graphs} and \autoref{subsec:rules-with-acs} we briefly reintroduce the framework of $\M$-adhesive categories and of rules with ACs; in  \autoref{subsec:critical-pairs} we reintroduce critical pairs for rules with ACs following \cite{EhrigHLOG10,EhrigGHLO12}; in \autoref{subsec:initial-conflicts} we reintroduce initial conflicts for plain rules, and in \autoref{sec:cond-conflict} we introduce initial parallel independent transformation pairs. This result is used in \autoref{sec:cond-conflict}, where we present the main results of the paper about \emph{initial conflicts for rules with ACs}. Then, in \autoref{sec:unfolding} we show our results on \emph{unfolding initial conflicts}. Finally, we conclude in \autoref{sec:conclusion} discussing some future work. The paper only includes proof sketches, detailed proofs can be found in the appendix





\section{Related Work}\label{sec:related-work}

Most work on checking \emph{confluence} for rule-based rewriting systems is based on the seminal paper from Knuth and Bendix~\cite{KB70}, who reduced the problem of checking local confluence to checking the joinability of a finite set of \emph{critical pairs} obtained from superposing or overlapping the left hand sides of pairs of rewriting rules. This technique has been extensively studied and applied in the area of term rewriting systems (see, for instance, \cite{Ohle02}), and it was introduced in the area of \emph{graph transformation} by Plump~\cite{Plump93,Plump94,Plump05} in the context of term-graph and hypergraph rewriting. Moreover, he also proved that (local) confluence of graph transformation systems is undecidable, even for terminating systems, as opposed to what happens in the area of term rewriting systems. However, recently, in \cite{BonchiGKSZ17} it is shown that confluence of terminating DPO transformation of graphs with interfaces is decidable. The authors explain that the reason is that interfaces play the same role as variables in term rewriting systems, where confluence is undecidable for terminating ground (i.e., without variables) systems, but decidable for non-ground ones. 

The notion of critical pairs in the area of graph transformation, as introduced by Plump~\cite{Plump93,Plump94,Plump05}, has the characteristic that their computation is exponential in the size of the preconditions of the rules. For this reason, different \emph{proper subsets of critical pairs} with a considerably reduced size were studied that are still complete~ \cite{LambersEO08,LambersBO0T18,AzziCR19}, clearing the way for a more efficient computation. The notion of \emph{essential critical pair}~\cite{LambersEO08} for graph transformation systems already allowed for a significant reduction, and, the notion of \emph{initial conflict}~\cite{LambersBO0T18}, introduced for the more general $\M$-adhesive systems, allowed for an even larger reduction. A problem with initial conflicts is that not all $\M$-adhesive categories necessarily have them. In this sense, in \cite{LambersBO0T18} it is shown that, in particular, typed graphs have initial conflicts and, a bit later, \cite{AzziCR19} extended that result proving that arbitrary $\M$-adhesive categories satisfying some given conditions also have initial conflicts. Moreover, they provided a simple way of constructing the initial pair of transformations for a given conflict.

A recent line of work concentrates on the development of \emph{multi-granular conflict detection techniques}~\cite{BornL0T17,Lambers0TBH18,LambersBKST19}. In particular, an extensive literature survey shows~\cite{Lambers0TBH18} that conflict detection is used at different levels of granularity depending on its application field. The overview shows that conflict detection can be used for the analysis and design phase of software systems (e.g. for finding inconsistencies in requirement specifications), for model-driven engineering (e.g. supporting model version management), for testing (e.g. generation of interesting test cases), or for optimizing rule-based computations (e.g. avoiding backtracking). These multi-granular techniques are presented for rules without application conditions (ACs). Our work builds further foundations for providing  multi-granular techniques also in the case of rules with ACs in the future. 

The use of (negative) \emph{application conditions} and of graph constraints, to limit the application of graph transformation rules, was introduced in \cite{Ehrig-Habel86a,HW95,HabelHT96}. Based on this notion of graph constraints, in \cite{Rensink04}, Rensink presented a logic for expressing graph properties, closely related to the logic of nested conditions of Habel and Penneman \cite{HabelP09}, shown to have the same expressive power as first-order logic on graphs, and being (refutationally) complete as demonstrated in Lambers and Orejas \cite{LOicgt14}.
Checking confluence for graph transformation systems with application conditions (ACs) has been studied in \cite{LambersEO06,Lambers2009} for the case of negative application conditions (NACs), and in \cite{EhrigHLOG10,EhrigGHLO12} for the more general case of  ACs. In the case of rules with ACs, it is an open issue to also come up with proper subsets of critical pairs of considerably reduced size (analogous to the previously mentioned works for rules without ACs). 


\section{Preliminaries}\label{sec:preliminaries} 

We start with a very brief introduction of $\M$-adhesive categories. We then revisit \emph{rules with nested application conditions (ACs)} (cf.~\autoref{subsec:rules-with-acs}) as well as the main parts of \emph{critical pair theory} for this type of rules~\cite{EhrigHLOG10,EhrigGHLO12} (cf.~\autoref{subsec:critical-pairs}). Thereafter, we reintroduce the notion of \emph{initial conflicts}~\cite{LambersBO0T18} for \emph{plain} rules, i.e. rules without nested application conditions (cf.~\autoref{subsec:initial-conflicts}). We also introduce the notion of \emph{initial parallel independent transformation pairs} as a counterpart (cf.~\autoref{subsec:parallel-independent}), since it will play a particular role when defining initial conflicts for rules with ACs in \autoref{subsec:initial-conflicts}. We assume that the reader is acquainted with the basic theory of DPO graph transformation and, in particular, the standard definitions of typed graphs and typed graph morphisms (see, e.g., \cite{EhrigEPT06}) and its associated category, {\bf Graphs$_{TG}$}.

\subsection{Graphs \& High-Level Structures}\label{subsec:graphs}



The results presented in this paper do not only apply to a specific class of graph transformation systems, like standard (typed) graph transformation systems, but to systems over any $\M$-adhesive category~\cite{EhrigGH10}. The idea behind the consideration of $\M$-adhesive categories is to avoid similar investigations for different instantiations like e.g. Petri nets, hypergraphs, and algebraic specifications. An $\M$-adhesive category is a category $\Cat$ with a distinguished morphism class $\M$ of monomorphisms satisfying certain properties. The most important one is the van Kampen (VK) property stating a certain kind of compatibility of pushouts and pullbacks along $\M$-morphisms. Moreover, additional properties are needed in our context: initial pushouts, describing the existence of a special ``smallest" pushout over a morphism, 
$\Epi'$-$\M$ pair factorizations, extending the classical epi-mono factorization to a pair of morphisms with the same codomain. The definitions of $\M$-adhesive categories, initial pushouts, $\Epi'$-$\M$ pair factorizations, and binary coproducts, can be found in~\cite{EhrigGHLO12,EhrigGHLO14}. The results in this paper require an $\M$-adhesive category where additional properties hold.

\begin{assumption}\label{ass:adhesive} We assume that $\tuple{\Cat,\M}$ is 
an $\M$-adhesive category with 
a unique $\Epi'$-$\M$ pair factorization (needed for \autoref{lem:shift}, \autoref{def:critical}, 
\autoref{thm:completeness}, \autoref{thm:ic-is-cp}, \autoref{cor:monotone-NACs}) and binary coproducts (needed for \autoref{lem:initial-parallel-independent}, \autoref{def:core-parallel}, \autoref{thm:completeness-core-parallel}). For the Local Confluence Theorem for initial conflicts of rules with ACs we in addition need initial pushouts (cf.~\autoref{subsec:completeness-confluence}).  
\end{assumption}


\begin{remark}[$\tuple{\Graphs_{TG},\M}$, $\tuple{\PTNets,\M}$, $\tuple{\Spec,\M_{strict}}$ are $\M$-adhesive and satisfy additional properties~\cite{EhrigEPT06,EhrigGH10}]\label{fac:adh} In particular, the category $\tuple{\Graphs_{TG},\M}$ with the class $\M$ of all injective typed graph morphisms is an $\M$-adhesive category. 
It has a unique $\Epi'$-$\M$ pair factorization where $\Epi'$ is the class of jointly surjective typed graph morphism pairs (i.e., the morphism pairs $(e_1,e_2)$ such that for each $x\in K$ there is a pre-image $a_1\in A_1$ with $e_1(a_1)=x$ or $a_2\in A_2$ with $e_2(a_2)=x$).  Binary coproduct objects correspond to disjoint unions of graphs. All other examples are also $\M$-adhesive categories and satisfy the additional properties for suitable choices of $\M$ and $\Epi'$. 
\end{remark}

\subsection{Rules with Application Conditions and Parallel Independence}\label{subsec:rules-with-acs}

We reintroduce nested application conditions~\cite{HabelP09} (in short, application conditions, or just ACs) following~\cite{EhrigGHLO12}. They generalize the corresponding notions in \cite{HabelHT96,KochMP05,EhrigEHP06}, where a negative (positive) application condition, short NAC (PAC), over a graph $P$, denoted $\neg \exists a$ ($\exists a$) is defined in terms of a morphism $a: P \rightarrow C$. Informally, a morphism $m: P \rightarrow G$ satisfies $\neg \exists a$ ($\exists a$) if there does not exist a morphism $q: C \rightarrow G$ extending $a$ to $m$ (if there exists $q$ extending $a$ to $m$). Then, an AC (also called \emph{nested AC}) is either the special condition $\ctrue$ or a pair of the form $\PE(a,\ac_C)$ or $\neg \PE(a,\ac_C)$, where the first case corresponds to a PAC and the second case to a NAC, and in both cases $\ac_C$ is an additional AC on $C$. Intuitively, a morphism $m: P \rightarrow G$ satisfies $\PE(a,\ac_C)$ if $m$ satisfies $a$ and the corresponding extension $q$ satisfies $\ac_C$. Moreover, ACs (and also NACs and PACs) may be combined with the usual logical connectors.

\begin{definition}[application condition and satisfaction]\label{def:condition}  An \emph{application condition} $\ac_P$ over an object $P$ is inductively defined as follows: 
\begin{itemize}
\item For every morphism $a\colon P\to C$ and every application condition $\ac_C$ over $C$, $\PE(a,\ac_C)$ is an application condition over $P$. 
\item For application conditions $c$, $c_i$ over P with $i\in I$ (for finite index sets $I$), $\neg c$ and $\wedge_{i \in I} c_i$ are application conditions over $P$. 
\end{itemize}
We define inductively when a morphism \emph{satisfies} an application condition:
\begin{itemize}
\item A morphism $p\colon P\to G$ satisfies an application condition $\PE(a,\ac_C)$, denoted $p\models \PE(a,\ac_C)$, if there exists an $\M$-morphism~$q$ such that $q\circ a = p$ and $q\models \ac_C$. 
\item A morphism $p\colon P\to G$ satisfies $\neg c$ if $p$ does not satisfy $c$ and satisfies $\wedge_{i \in I} c_i$ if it satisfies each $c_i$ ($i \in I$). 
\end{itemize}
\end{definition}
\[\tikz[node distance=2em,shape=rectangle,outer sep=1pt,inner sep=2pt]
{
\node(P){$P$};
\node(G)[strictly below right of=P]{$G$};
\node(C)[strictly above right of=G]{$C,$};
\draw[morphism] (P) -- node[overlay,above](a){$a$} (C);
\draw[morphism] (P) -- node[overlay,below left]{$p$} (G);
\draw[altmonomorphism] (C) -- node[overlay,below right](q){$q$} (G);
\draw[draw=white] (a) -- node[overlay](tr1){=} (G);
\node(c)[outer sep=0pt,inner sep=0pt,node distance=0em,strictly right of=C]
{\tikz[draw=black,fill=lightgray]
{
\filldraw (0,0) -- (0.8,0.2) -- node
{\small \;$\ac_C$} (0.8,0) -- (0.8,-0.2) -- (0,0);
}
};
\draw[draw=white] (q) -- node[overlay,sloped](tr1){$\models$} (c);
\node(Y)[node distance=0.1em,strictly right of=c]{$)$};
\node(X)[node distance=0.1em,strictly left of=P]{$\PE($}
;}
\]

Note that the empty conjunction (equiv. to $\ctrue$), satisfied by each morphism, serves as base case in the inductive definition.  Moreover, $\PE a$ (resp. $\PA(a,\ac_C)$)    abbreviates $\PE(a,\ctrue)$ (resp. $\neg\PE(a,\neg\ac_C)$).

ACs are used to restrict the application of rules to a given object. The idea is to equip the precondition (or left hand side) of rules with an application condition\footnote{We could have also allowed to equip the right-hand side of rules with an additional AC, but this case can be reduced to rules with left ACs only as shown in \autoref{lem:left}.}. Then we can only apply a given rule to an object $G$ if the corresponding match morphism satisfies the AC of the rule. However, for technical reasons\footnote{For example, symbolic transformation pairs as introduced later, or also critical pairs for rules with ACs (see \autoref{def:critical}) consist of transformations that do not need to satisfy the associated ACs.}, we also introduce the application of rules \emph{disregarding} the associated ACs.

\begin{definition}[rules and transformations]\label{def:rules} A \emph{rule} $\prule = \tuple{p,\ac_L}$ consists of a \emph{plain} rule $p=\crule{L}{I}{R}$ with $I\injto L$ and $I\injto R$ morphisms in $\M$ and an application condition $\ac_L$ over $L$. 
\[\tikz[node distance=2em,shape=rectangle,outer sep=1pt,inner sep=2pt,label distance=-1.25em]{
\node(L){$L$};
\node(K)[strictly right of=L]{$I$};
\node(R)[strictly right of=K]{$R$};
\node(D)[strictly below of=K]{$D$};
\node(G)[strictly below of=L]{$G$};
\node(H)[strictly below of=R]{$H$};
\draw[altmonomorphism] (K) -- node[overlay,above]{} (L);
\draw[monomorphism] (K) -- node[overlay,above]{} (R);
\draw[altmonomorphism] (D) -- (G);
\draw[monomorphism] (D) -- (H);
\draw[morphism] (L) -- node[overlay,left](m){$m$} (G);
\draw[morphism] (K) -- (D);
\draw[morphism] (R) -- node[overlay,right](m*){$m^*$}(H);
\draw[draw=none] (L) -- node[overlay](po1){(1)} (D);
\draw[draw=none] (R) -- node[overlay](po2){(2)} (D);
\node(acL)[outer sep=0pt,inner sep=0pt,node distance=0em,strictly left of=L]{
\tikz[baseline,draw=black,fill=lightgray]{\filldraw (0,0) -- node[left,pos=0.79,overlay,outer sep=2ex](acL2){\small $\ac_L$} (-0.8,0.12) -- (-0.8,-0.12) -- (0,0);}};
\draw[draw=none] (m) -- node[overlay,sloped](tr1){$\mathrel{=}\joinrel\mathrel{|}$} (acL);
}\]
A \emph{direct transformation} \/ $t: G\dder_{\prule,m,m^*} H$\/  consists of two pushouts (1) and~(2), called DPO, with match $m$ and comatch $m^*$ such that $m\models\ac_L$. $G\hookleftarrow D \hookrightarrow H$ is called the \emph{derived span} of $t$. An \emph{AC-disregarding direct transformation}\/ $G\dder_{\prule,m,m^*} H$\/ consists of DPO (1) and (2), where $m$ does not necessarily need to satisfy $\ac_L$. Given a set of rules $\mathcal{R}$ for $\tuple{\Cat,\M}$, the triple $\tuple{\Cat,\M, \R}$ is an \emph{$\mathcal{M}$-adhesive system}. \end{definition}

\begin{remark} In the rest of the paper we assume that each rule  (resp. transformation or $\mathcal{M}$-adhesive system) comes with ACs. Otherwise, we state that we have a \emph{plain} rule (resp. transformation or $\mathcal{M}$-adhesive system). This plain case can also be seen as a special case of a rule (resp. transformation or $\mathcal{M}$-adhesive system) with ACs in the sense that the ACs are (equivalent to) $\true$. 
\end{remark}

ACs can be shifted over morphisms and rules (from right to left and vice versa) as shown in the following lemma (for constructions see \cite{EhrigGHLO14}~\footnote{Since this construction entails the enumeration of jointly epimorphic morphism pairs, its computation has exponential complexity in the size of the precondition of the rule and the size of the AC.} and \cite{HabelP09,EhrigGHLO14}, respectively). We only describe the right to left case in~\autoref{lem:left}, since the left to right case is symmetrical. 

\begin{lemma}[shift ACs over morphisms~\cite{EhrigGHLO14}]\label{lem:shift} 
There is a transformation $\Shift$ from morphisms and ACs to ACs such that for each AC, $\ac_P$, and each morphism $b\colon P \to P'$, $\Shift$ transforms $\ac_P$ via $b$ into an AC $\Shift(b,\ac_P)$ over $P'$ such that for each morphism $n\colon P'\to H$ it holds that $n\circ b\models \ac_P \shortiff n\models \Shift(b,\ac_P)$.
\end{lemma}

\begin{lemma}[shift ACs over rules \cite{HabelP09,EhrigGHLO14}]\label{lem:left} 
There is a transformation $\Left$ from rules and ACs to ACs such that for every rule $\prule: L \hookleftarrow I \hookrightarrow R$ and every AC on $R$, $\ac_R$, $\Left$ transforms $\ac_R$ via $\prule$ into the AC $\Left(\prule,\ac_R)$ on $L$, such that for every direct transformation $G\dder_{\prule,m,m^*}H$,  $m\models\Left(\prule,\ac_R)\shortiff m^*\models\ac_R$.
\end{lemma}


For \emph{parallel independence}, when working with rules with ACs, we need not only that each rule does not delete any element  which is part of the match of the other rule, but also that the resulting transformation defined by each rule application still satisfies the ACs of the other rule application. 

\begin{definition}[transformation pairs and parallel independence]\label{def:parallelIndependence}
A \emph{transformation pair} $H_1\ldder_{\prule_1,o_1}G\dder_{\prule_2,o_2} H_2$ is \emph{parallel
independent} if there exists a morphism $d_{12}\colon L_1\to D_2$ such that $k_2\circ d_{12}=o_1$ and $c_{2}\circ d_{12}\models \ac_{L_1}$ and there exists a morphism $d_{21}\colon L_2\to D_1$ such that $k_1\circ d_{21}=o_2$ and $c_{1}\circ d_{21}\models \ac_{L_2}$.

\[\tikz[node distance=4em,shape=rectangle,outer sep=1pt,inner sep=2pt,label distance=-1.25em]{
\node(H){$G$};
\node(D1)[node distance=7em,left of=H]{$D_1$};
\node(G)[left of=D1]{$H_1$};
\node(L1)[above of=G]{$R_1$};
\node(K1)[right of=L1]{$I_1$};
\node(R1)[right of=K1]{$L_1$};
\node(D2)[node distance=7em,right of=H]{$D_2$};
\node(M)[right of=D2]{$H_2$};
\node(R2)[above of=M]{$R_2$};
\node(K2)[left of=R2]{$I_2$};
\node(L2)[left of=K2]{$L_2$};
\draw[altmonomorphism] (K1) -- (L1);
\draw[monomorphism] (K1) -- (R1);
\draw[monomorphism] (D1) -- node[below]{\small $k_1$} (H);
\draw[altmonomorphism] (D1) --  node[below]{\small $c_1$} (G);
\draw[morphism] (L1) -- node[left]{} (G);
\draw[morphism] (K1) -- (D1);
\draw[morphism] (R1) -- node[below=5pt]{\small $o_1$} (H);
\draw[altmonomorphism] (K2) -- (L2);
\draw[monomorphism] (K2) -- (R2);
\draw[altmonomorphism] (D2) -- node[below]{\small $k_2$} (H);
\draw[monomorphism] (D2) -- node[below]{\small $c_2$} (M);
\draw[morphism] (L2) -- node[below=5pt]{\small $o_2$} (H);
\draw[morphism] (K2) -- (D2);
\draw[morphism] (R2) -- (M);
\draw[morphism,dotted] (L2) -- node[left=15pt]{\small $d_{21}$} (D1);
\draw[morphism,dotted] (R1) -- node[right=15pt]{\small $d_{12}$} (D2);
\node(acL1)[outer sep=0pt,inner sep=0pt,node distance=0em,strictly right of=R1]{
\tikz[baseline,draw=black,fill=lightgray]{\filldraw (0,0) -- node[above,pos=0.6,overlay,outer sep=1ex](acL1){\small $\ac_{L_1}$} (0.3,0.12) -- (0.3,-0.12) -- (0,0);}};
\node(acL2)[outer sep=0pt,inner sep=0pt,node distance=1em,left of=L2]{\tikz[baseline,draw=black,fill=lightgray]{\filldraw (0,0) -- node[above,pos=0.6,overlay,outer sep=1ex](acL2){\small $\ac_{L_2}$} (-0.3,0.12) -- (-0.3,-0.12) -- (0,0);}};}\]
\end{definition}

We say that a transformation pair is \emph{in conflict} or \emph{conflicting} if it is parallel dependent.  We distinguish different conflict types, generalizing straightforwardly the conflict characterization introduced for rules with NACs~\cite{LambersEO06}.  The transformation pair $H_1\ldder_{\prule_1,o_1}G\dder_{\prule_2,o_2} H_2$ is a \emph{use-delete} (resp. \emph{delete-use}) conflict if in \autoref{def:parallelIndependence} the commuting morphism $d_{12}$ (resp. $d_{21}$) does not exist, i.e. the second (resp. first) rule deletes something used by the first (resp. second) one. Moreover, it is an \emph{AC-produce} (resp. \emph{produce-AC}) conflict if in \autoref{def:parallelIndependence} the commuting morphism $d_{12}$ (resp. $d_{21}$) exists, but an extended match is produced by the second (resp. first) rule that does not satisfy the rule AC of the first (resp. second) rule. If a transformation pair is an \emph{AC-produce} or \emph{produce-AC} conflict, then we also say that it is an \emph{AC conflict} or \emph{AC conflicting}.

\begin{remark}[use-delete XOR AC-produce] A use-delete (resp. delete-use) conflict  cannot occur simultaneously to an AC-produce (resp. produce-AC) conflict. This is because the AC of the first (resp. second) rule can only be violated iff there exists an extended match for the first (resp. second) rule. However, a use-delete (resp. delete-use) conflict may occur simultaneously to a produce-AC (resp. AC-produce) conflict, since in this case the extended match for the first (resp. second) rule does not exist, whereas the extended match for the second (resp. first) rule exists and violates the AC, i.e. both conflict types occur on opposite sides of the diagram in \autoref{def:parallelIndependence}.
\end{remark}

For grasping the notion of completeness of transformation pairs w.r.t. a property like parallel (in-)dependence, it is first important to understand how a given transformation can be extended to another transformation. In particular, an \emph{extension diagram} describes how a transformation $t\colon G_0\der G_n$ can be extended to a transformation $t'\colon G'_0\der G'_n$ via the same rules and an \emph{extension morphism} $k_0\colon G_0\to G'_0$ that maps $G_0$ to $G'_0$ as shown in the following diagram on the left. For each rule application and transformation step, we have two double pushout diagrams  as shown on the right, where the rule $\rho_{i+1}$ is applied to both $G_i$ and $G_i'$. 

\[\tikz[node distance=2em,shape=rectangle,outer sep=1pt,inner sep=2pt,label distance=-1.25em]{
\node(K){$G_0$};
\node(R)[node distance=3em,strictly right of=K]{$G_n$};
\node(D)[strictly below of=K]{$G'_0$};
\node(H)[strictly below of=R]{$G'_n$};
\node(LL)[node distance=3em,strictly right of=R]{$L_{i+1}$};
\node(KK)[node distance=3em,strictly right of=LL]{$I_{i+1}$};
\node(RR)[node distance=3em,strictly right of=KK]{$R_{i+1}$};
\node(GL)[strictly below of=LL]{$G_{i}$};
\node(GK)[strictly below of=KK]{$D_{i}$};
\node(GR)[strictly below of=RR]{$G_{i+1}$};
\node(GGL)[strictly below of=GL]{$G_{i}'$};
\node(GGK)[strictly below of=GK]{$D_{i}'$};
\node(GGR)[strictly below of=GR]{$G_{i+1}'$};
\draw[morphism] (KK) --  (LL);
\draw[morphism] (KK) --  (RR);
\draw[morphism] (GK) --  (GL);
\draw[morphism] (GK) --  (GR);
\draw[morphism] (GGK) --  (GGL);
\draw[morphism] (GGK) --  (GGR);
\draw[morphism] (LL) --  (GL);
\draw[morphism] (GL) --  (GGL);
\draw[morphism] (KK) --  (GK);
\draw[morphism] (GK) --  (GGK);
\draw[morphism] (RR) --  (GR);
\draw[morphism] (GR) --  (GGR);
\draw[morphism] (K) -- node[overlay,left](m){\small $k_0$} (D);
\draw[morphism] (R) -- node[overlay,right](h){\small $k_n$} (H);
\draw[derivation] (K) -- node[overlay,above](h){\small $*$} (R);
\draw[derivation] (D) -- node[overlay,above](h){\small $*$} (H);
\draw[draw=none] (K) -- node[overlay](po1){\small (1)} (H);}\]

We introduce two different notions of completeness,  distinguishing $\mathcal{M}$-completeness from regular completeness, depending on the membership of the extension morphism in $\mathcal{M}$. It is known that critical pairs (resp. initial conflicts) for \emph{plain rules} are $\mathcal{M}$-complete (resp. complete) w.r.t. parallel dependence~\cite{EhrigEPT06,LambersBO0T18}.  In \autoref{subsec:critical-pairs}, we reintroduce the fact that critical pairs for rules with ACs are $\mathcal{M}$-complete w.r.t. parallel dependence, but as symbolic transformation pairs. We learn in \autoref{sec:cond-conflict} that initial conflicts for rules with ACs are also complete in this symbolic way. 

\begin{definition}[($\mathcal{M}$-)completeness of transformation pairs]\label{def:completeness-concrete}
A set of \emph{transformation pairs} $\mathcal{S}$ for a pair of rules $\tuple{\prule_1,\prule_2}$ is \emph{complete} (resp. \emph{$\mathcal{M}$-complete}) w.r.t. parallel (in-)dependence if there is a pair $P_1\ldder_{\prule_1,o_1}K\dder_{\prule_2,o_2} P_2$ from $\mathcal{S}$ and an extension diagram via extension morphism $m$ (resp. $m\in\M$) for each parallel (in-)dependent direct transformation pair $H_1\ldder_{\prule_1,m_1}G\dder_{\prule_2,m_2} H_2$ .
\begin{figure}
\centering
\begin{tikzpicture}
[node distance=3em,shape=rectangle,outer sep=1pt,inner sep=2pt,label distance=-1.25em]{
\node(K){$K$};
\node(P1)[strictly left of=K]{$P_1$};
\node(P2)[strictly right of=K]{$P_2$};
\node(G)[node distance=2em,strictly below of=K]{$G$};
\node(H1)[strictly left of=G]{$H_1$};
\node(H2)[strictly right of=G]{$H_2$};
\draw[morphism] (K) -- node[overlay,right](m){\small $m$} (G);
\draw[morphism] (P1) -- node[overlay,left](z1){}(H1);
\draw[morphism] (P2) -- node[overlay,right](z2){}(H2);
\draw[derivation] (K) -- node[overlay,above](h2){\small $\prule_1,o_1$}(P1);
\draw[derivation] (K) -- node[overlay,above](h2){\small $\prule_2,o_2$}(P2);
\draw[derivation] (G) -- node[overlay,below](h2){\small $\prule_1,m_1$}(H1);
\draw[derivation] (G) -- node[overlay,below](h2){\small $\prule_2,m_2$}(H2);}
\end{tikzpicture}
\caption{($\mathcal{M}$-)completeness of transformation pairs}\label{fig:completeness}
\end{figure}
\end{definition}

\subsection{Critical Pairs}\label{subsec:critical-pairs}

Critical pairs for plain rules are just transformation pairs, where morphisms $o_1$ and $o_2$ are in $\Epi'$ (i.e., roughly, $K$ is an overlapping of $L_1$ and $L_2$). In the category of  {\bf Graphs} they lead to finite and complete subsets of finite conflicts~\cite{EhrigEHP06} (assumed that the rule graphs are also finite). However, when rules include ACs, we cannot use the same notion of critical pair since, as we show in  \autoref{thm:infinite-set-complete}, in general, for any two rules with ACs, there is no complete set of transformation pairs that is finite. To avoid this problem, our critical pairs for rules with ACs also include ACs, as in ~\cite{EhrigHLOG10,EhrigGHLO12}, where they are proved to be \emph{$\mathcal{M}$-complete}, and they are also finite in the category of {\bf Graphs} (assumed again that the rules are finite). Moreover, the converse property also holds, in the sense that every critical pair can be extended to a conflict (or pair of parallel dependent rule applications). 


In particular, critical pairs are based on the notion of \emph{symbolic transformation pairs}, which are pairs of \emph{AC-disregarding transformations} on some object $K$ with two special ACs on $K$. These two ACs, $\ac_K$ (\emph{extension AC}) and $\ac^*_K$ (\emph{conflict-inducing AC}), are used to characterize which embeddings of this pair, via some morphism $m:K\to G$, give rise to a transformation pair that is parallel dependent. If $m\models\ac_K$, then $m\circ o_1: L_1 \rightarrow G$ and $m\circ o_2: L_2 \rightarrow G$ are two morphisms,  
 satisfying the associated ACs of $\rho_1$ and $\rho_2$, respectively. Moreover, if $m\models\ac^*_K$, then the two transformations $H_1\ldder_{\prule_1,m\circ o_1}G\dder_{\prule_2,m\circ o_2} H_2$ are parallel dependent. 
Symbolic transformation pairs allow us to present critical pairs as well as initial conflicts (cf.~\autoref{subsec:initial-conflicts}) in a compact and unified way, since they both are instances of symbolic transformation pairs. Finally, note that each symbolic transformation pair $stp_K: \tuple{tp_K,ac_K,ac_K^*}$ is by definition uniquely determined (up to isomorphism and equivalence of the extension AC and conflict-inducing AC) by its underlying AC-disregarding transformation pair. 

\begin{definition}[symbolic transformation pair]\label{def:symbolic-pair} Given rules $\prule_1=\tuple{p_1,\ac_{L_1}}$ and $\prule_2=\tuple{p_2,\ac_{L_2}}$, a \emph{symbolic transformation pair} $stp_K: \tuple{tp_K,ac_K,ac_K^*}$ for $\tuple{\prule_1,\prule_2}$ consists of a pair $tp_K: P_1\ldder_{\prule_1,o_1} K\dder_{\prule_2,o_2} P_2$ of AC-disregarding transformations together with  
ACs $\ac_K$ and $\ac^*_{K}$ on $K$ given by: 
\\
\indent $\ac_K=\Shift(o_1,\ac_{L_1})\wedge\Shift(o_2,\ac_{L_2})$, called \emph{extension AC}, and 
\\ 
\indent $\ac^*_K= \neg(\ac^*_{K,d_{12}}\wedge\ac^*_{K,d_{21}})$, called \emph{conflict-inducing AC}
\\ 
with $\ac^*_{K,d_{12}}$ and $\ac^*_{K,d_{21}}$ given as follows:

\[\begin{array}{l}
\mbox{ if }(\exists\;d_{12}\mbox{ with }k_2{\circ}d_{12}{=}o_1) 
\hspace{4pt} \mbox{ then }  \ac^*_{K,d_{12}}=\Left(p^*_2,\Shift(c_2{\circ}d_{12},\ac_{L_1})) \\
\hspace{118pt} \mbox{ else } \ac^*_{K,d_{12}} = \cfalse
\\
\mbox{ if }(\exists\;d_{21}\mbox{ with }k_1{\circ}d_{21}{=}o_2) 
\hspace{4pt} \mbox{ then } \ac^*_{K,d_{21}}=\Left(p^*_1,\Shift(c_1{\circ}d_{21},\ac_{L_2})) \\
\hspace{118pt} \mbox{ else } \ac^*_{K,d_{21}} = \cfalse
\end{array}\] 
where $p^*_1=\nrule{K}{D_1}{P_1}{k_1}{c_1}$ and $p^*_2=\nrule{K}{D_2}{P_2}{k_2}{c_2}$ are defined by the corresponding double pushouts.\\ 

\[\tikz[node distance=4em,shape=rectangle,outer sep=1pt,inner sep=2pt,label distance=-1.25em]{
\node(H){$K$};
\node(D1)[node distance=7em,left of=H]{$D_1$};
\node(G)[left of=D1]{$P_1$};\node(G')[node distance=2em,left of=G]{$p^*_1:$}; 
\node(L1)[above of=G]{$R_1$};\node(L1')[node distance=2em,left of=L1]{$p_1:$};
\node(K1)[right of=L1]{$I_1$};
\node(R1)[right of=K1]{$L_1$};
\node(D2)[node distance=7em,right of=H]{$D_2$};
\node(M)[right of=D2]{$P_2$};\node(M')[node distance=2em,right of=M]{$:p^*_2$};
\node(R2)[above of=M]{$R_2$};\node(R2')[node distance=2em,right of=R2]{$:p_2$};
\node(K2)[left of=R2]{$I_2$};
\node(L2)[left of=K2]{$L_2$};
\draw[altmonomorphism] (K1) -- (L1);
\draw[monomorphism] (K1) -- (R1);
\draw[monomorphism] (D1) -- node[below]{\small $k_1$} (H);
\draw[altmonomorphism] (D1) --  node[below]{\small $c_1$} (G);
\draw[morphism] (L1) -- node[left]{} (G);
\draw[morphism] (K1) -- (D1);
\draw[morphism] (R1) -- node[below = 5pt]{\small $o_1$} (H);
\draw[altmonomorphism] (K2) -- (L2);
\draw[monomorphism] (K2) -- (R2);
\draw[altmonomorphism] (D2) -- node[below]{\small $k_2$} (H);
\draw[monomorphism] (D2) -- node[below]{\small $c_2$} (M);
\draw[morphism] (L2) -- node[below=5pt]{\small $o_2$} (H);
\draw[morphism] (K2) -- (D2);
\draw[morphism] (R2) -- (M);
\draw[morphism,dotted] (L2) -- node[left=15pt]{\small $d_{21}$} (D1);
\draw[morphism,dotted] (R1) -- node[right=15pt]{\small $d_{12}$} (D2);
\node(acL1)[outer sep=0pt,inner sep=0pt,node distance=0em,strictly right of=R1]{
\tikz[baseline,draw=black,fill=lightgray]{\filldraw (0,0) -- node[above,pos=0.6,overlay,outer sep=1ex](acL1){\small $\ac_{L_1}$} (0.3,0.12) -- (0.3,-0.12) -- (0,0);}};
\node(acL2)[outer sep=0pt,inner sep=0pt,node distance=1em,left of=L2]{\tikz[baseline,draw=black,fill=lightgray]{\filldraw (0,0) -- node[above,pos=0.6,overlay,outer sep=1ex](acL2){\small $\ac_{L_2}$} (-0.3,0.12) -- (-0.3,-0.12) -- (0,0);}};}\]
\end{definition}

A \emph{critical pair}\footnote{A symbolic transformation pair with matches belonging to $\Epi'$ is called a weak critical pair in~\cite{EhrigHLOG10,EhrigGHLO12}} is now a symbolic transformation pair in a minimal context such that there exists at least one extension to a pair of transformations being parallel dependent.

\begin{definition}[critical pair]\label{def:critical} Given rules $\prule_1=\tuple{p_1,\ac_{L_1}}$ and $\prule_2=\tuple{p_2,\ac_{L_2}}$, a \emph{critical pair} for $\tuple{\prule_1,\prule_2}$ is a symbolic transformation pair $stp_K: \tuple{tp_K,\ac_K,\ac^*_K}$, where the match pair $(o_1,o_2)$ of  $tp_K$ is in $\Epi'$, and there exists a morphism $m\colon K\to G \in \M$ 
such that $m\models\ac_K\wedge\ac^*_{K}$ and  $m_i=m \circ o_i$, for $i=1,2$, satisfy the gluing conditions, i.e.\;$m_i$ has a pushout complement w.r.t. $p_i$.\end{definition}


\begin{definition}[($\mathcal{M}$-)completeness of symbolic transformation pairs]\label{def:completeness-symbolic}
A set of symbolic transformation pairs $\mathcal{S}$ for a pair of rules $\tuple{\prule_1,\prule_2}$ is \emph{complete} (resp. \emph{$\mathcal{M}$-complete}) w.r.t. parallel dependence if there is a symbolic transformation pair $stp_K: \tuple{tp_K: P_1\ldder_{\prule_1,o_1}K\dder_{\prule_2,o_2} P_2,\ac_K,\ac^*_K}$ from $\mathcal{S}$ and an extension diagram as depicted in \autoref{fig:completeness} with $m: K\rightarrow G$ (resp. $m: K\rightarrow G \in \mathcal{M}$) and $m\models\ac_K\wedge\ac^*_K$ for each parallel dependent direct transformation $H_1\ldder_{\prule_1,m_1}G\dder_{\prule_2,m_2} H_2$.
\end{definition}

\begin{theorem}[$\mathcal{M}$-completeness of critical pairs~\cite{EhrigHLOG10,EhrigGHLO12}]\label{thm:completeness} The set of \emph{critical pairs} for a pair of rules $\tuple{\prule_1,\prule_2}$ is \emph{$\mathcal{M}$-complete} w.r.t. parallel dependence. Moreover, for each critical pair $P_1\ldder_{\prule_1,o_1}K\dder_{\prule_2,o_2} P_2$ for $\tuple{\prule_1,\prule_2}$ there is a parallel dependent pair $H_1\ldder_{\prule_1,m_1}G\dder_{\prule_2,m_2} H_2$ and a morphism $m\colon K\to G\in\M$ such that $m\models\ac_K\wedge\ac^*_K$ leading to the above extension diagram.
\end{theorem}

Note that based on $\mathcal{M}$-completeness it is possible to formulate also a Local Confluence Theorem for critical pairs of rules with ACs for $\mathcal{M}$-adhesive categories with $\mathcal{M}$-initial pushouts~\cite{EhrigHLOG10,EhrigGHLO12}.


\subsection{Initial Conflicts for Plain Rules}\label{subsec:initial-conflicts}
\emph{Initial conflicts} for plain rules follow an alternative approach to the original idea of critical pairs. Instead of considering all conflicting transformations in a minimal context (materialized by a pair of jointly epimorphic matches), 
initial conflicts use the notion of \emph{initiality of transformation pairs} to obtain a more declarative view on the minimal context of critical pairs. Each initial conflict is a critical pair but not the other way round. Moreover, all initial conflicts for plain rules are complete w.r.t. parallel dependence and they still satisfy the Local Confluence Theorem for plain rules. Consequently, initial conflicts for plain rules represent an important, proper   
subset of critical pairs for performing static conflict detection as well as local confluence analysis. 
The contribution of this paper is to demonstrate how to achieve a similar situation for rules with ACs. 

\begin{definition}[initial transformation pair]
\label{def:ini_pair}
Given a pair of plain direct transformations $tp: H_1\ldder_{p_1,m_1}G\dder_{p_2,m_2} H_2$, then $tp^I: H^I_1\ldder_{p_1,m^I_1}G^I \dder_{p_2,m^I_2} H^I_2$ is an \emph{initial  transformation pair} for $tp$ if it can be embedded into $tp$ via extension diagrams (1) and (2) and extension morphism $f^I$ as in \autoref{fig:initial-tp} such that for each  transformation pair $tp': H'_1\ldder_{p_1,m'_1}G'\dder_{p_2,m'_2} H'_2$ that can be embedded into $tp$ via extension diagrams (3) and (4) and extension morphism $f$ as in \autoref{fig:initial-tp} it holds that $tp^I$ can be embedded into $tp'$ via unique extension diagrams (5) and (6) and unique vertical morphism $f'^I$ s.t. $f \circ f'^I = f^I$. 
\end{definition}

\begin{figure}[h]
\[
\xymatrix{
H^I_1 \ar[d]_{g^I_1} \ar@{}|{(1)}[dr] & G^I \ar@{=>}[l]_(.5){p_1,m^I_1}  \ar@{=>}[r]^(.5){p_2,m^I_2} \ar@{}|{(2)}[dr] \ar[d]_{f^I}    & H^I_2 \ar[d]^{g^I_2}\\
H_1                         & G \ar@{=>}[l] _(.5){p_1,m_1} \ar@{=>}[r] ^(.5){p_2,m_1} & H_2
}
\vspace{.5cm}
\xymatrix{
H^I_1 \ar[d]_{g'^I_1} \ar@{}|{(5)}[dr] & G^I \ar@{=>}[l]_(.5){p_1,m^I_1}  \ar@{=>}[r]^(.5){p_2,m^I_2} \ar@{}|{(6)}[dr] \ar[d]_{f'^I}    & H^I_2 \ar[d]^{g'^I_2}\\
H'_1 \ar[d]_{g_1} \ar@{}|{(3)}[dr] & G' \ar@{=>}[l]_(.5){p_1,m\rq{}_1}  \ar@{=>}[r]^(.5){p_2,m\rq{}_2} \ar@{}|{(4)}[dr] \ar[d]_{f}    & H'_2 \ar[d]^{g_2}\\
H_1                         & G \ar@{=>}[l] _(.5){p_1,m_1} \ar@{=>}[r] ^(.5){p_2,m_2}                             & H_2
}
\]
\vspace{-.4cm}
\caption{Initial transformation pair $H^I_1\ldder_{p_1,m^I_1}G^I \dder_{p_2,m^I_2} H^I_2$ for $H_1\ldder_{p_1,m_1}G\dder_{p_2,m_2} H_2$ 
}
\label{fig:initial-tp}
\end{figure}

As shown in~\cite{LambersBO0T18} an initial transformation pair is \emph{unique} up to isomorphism w.r.t. a given transformation pair for plain rules.  The notion of initial conflicts is based on the requirement of the  \emph{existence of initial transformation pairs} for parallel dependent or \emph{conflicting} plain transformation pairs. 
Note that for the category of typed graphs, it is shown in~\cite{LambersBO0T18} that this requirement holds. Moreover, \cite{AzziCR19} extended that result proving that arbitrary $\M$-adhesive categories fulfilling some extra conditions also satisfy it. 

\begin{definition}[existence of initial transformation pair for conflict]\label{def:existence-initial}
A plain $\mathcal{M}$-adhesive system has \emph{initial transformation pairs for conflicts} if, for each transformation pair $tp$ in conflict, the initial transformation pair $tp^I$ exists. 
\end{definition} 

Now initial conflicts for plain rules represent the set of all possible ``smallest'' conflicts. It is shown in~\cite{LambersBO0T18} that for a plain $\mathcal{M}$-adhesive system each initial conflict is a special critical pair. 

\begin{definition}[initial conflict]\label{def:core}
Given a plain $\mathcal{M}$-adhesive system with initial transformation pairs for conflicts, a pair of direct transformations in conflict $tp: H_1\ldder_{p_1,m_1}G\dder_{p_2,m_2} H_2$ is an \emph{initial conflict} if it is isomorphic to the initial transformation pair $tp^I$ for $tp$.
\end{definition}

Initial conflicts for plain rules are complete as transformation pairs w.r.t. parallel dependence, whereas critical pairs for plain rules are $\mathcal{M}$-complete~\cite{EhrigEHP06}.  

\begin{theorem}[completeness of initial conflicts~\cite{LambersBO0T18}]
\label{thm:completeness-core}
Consider a plain $\mathcal{M}$-adhesive system with initial transformation pairs for conflicts. The set of \emph{initial conflicts} for a pair of plain rules $\tuple{p_1,p_2}$ is \emph{complete} w.r.t. parallel dependence. 
\end{theorem}

Initial conflicts for plain rules are moreover minimally complete as transformation pairs w.r.t. parallel dependence, i.e. there does not exist any smaller set, which is also complete.

\begin{corollary}[minimally complete]\label{cor:minimal}
Consider a plain $\mathcal{M}$-adhesive system with initial transformation pairs for conflicts. The set of initial conflicts $\mathcal{S}$ (up to isomorphism) for a pair of plain rules $\tuple{p_1,p_2}$ is \emph{minimally complete} w.r.t. parallel dependence, i.e. there does not exist any smaller set $\mathcal{S'}$ of conflicts for $\tuple{p_1,p_2}$ that is complete w.r.t. parallel dependence.
\end{corollary}

The Local Confluence Theorem (requiring initial POs) can be formulated for initial conflicts of plain rules~\cite{LambersBO0T18} similarly to the one for classical critical pairs (for plain rules)~\cite{EhrigEPT06}.
%

\subsection{Initial Parallel Independent Transformation Pairs for Plain Rules}\label{subsec:parallel-independent}

In this section, we show the existence of initial transformation pairs for \emph{parallel independent transformation pairs}, allowing us to define a \emph{complete} subset also w.r.t. parallel independence. The proof requires the existence of binary coproducts.

\begin{lemma}[existence of initial transformation pair for parallel independent transformation pair]
\label{lem:initial-parallel-independent} 
Given a pair of parallel independent plain direct transformations $tp: H_1\ldder_{p_1,m_1}G\dder_{p_2,m_2} H_2$, then $tp_{L_1 +L_2}: R_1 + L_2 \ldder_{p_1, i_1} L_1 + L_2 \dder_{p_2, i_2} L_1 + R_2$, where $i_1: L_1 \rightarrow L_1 + L_2$ and $i_2: L_2 \rightarrow L_1 + L_2$ are the coproduct morphisms, is initial for $tp$.
\end{lemma}
\begin{proof}[idea]
The key issue in this proof is to show that, if $m: L_1 + L_2 \rightarrow G$ is the mediating morphism for $m_1: L_1 \rightarrow G$ and  $m_2: L_2 \rightarrow G$, then $m$ defines the extension diagrams in Fig. \ref{fig:embedding-pi}. \qed
\end{proof}

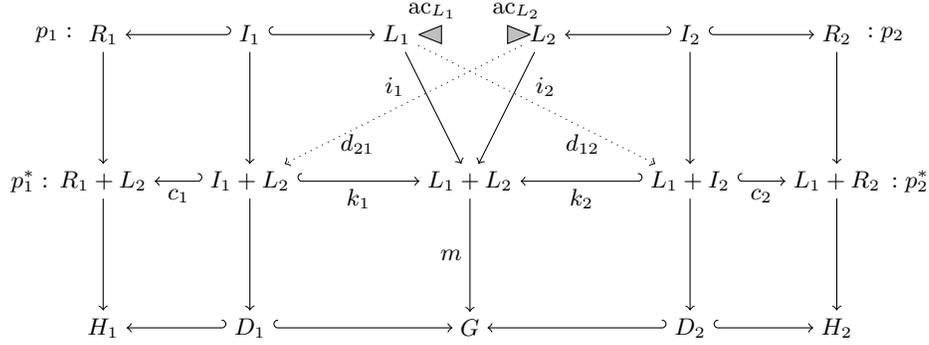
\begin{figure}
\centering
\begin{tikzpicture}[node distance=6em,shape=rectangle,outer sep=1pt,inner sep=2pt,label distance=-1.25em]{
\node(H){$L_1 + L_2$};
\node(D1)[node distance=9em,left of=H]{$I_1 + L_2$};
\node(G)[left of=D1]{$R_1 + L_2$};\node(G')[node distance=3em,left of=G]{$p^*_1:$}; 
\node(L1)[above of=G]{$R_1$};\node(L1')[node distance=2em,left of=L1]{$p_1:$};
\node(K1)[right of=L1]{$I_1$};
\node(R1)[right of=K1]{$L_1$};
\node(D2)[node distance=9em,right of=H]{$L_1 + I_2$};
\node(M)[right of=D2]{$L_1 + R_2$};\node(M')[node distance=3em,right of=M]{$:p^*_2$};
\node(R2)[above of=M]{$R_2$};\node(R2')[node distance=2em,right of=R2]{$:p_2$};
\node(K2)[left of=R2]{$I_2$};
\node(L2)[left of=K2]{$L_2$};
\node(G2)[below of=H]{$G$};
\node(D1P)[below of=D1]{$D_1$};
\node(H1)[below of=G]{$H_1$};
\node(D2P)[below of=D2]{$D_2$};
\node(H2)[below of=M]{$H_2$};
\draw[altmonomorphism] (K1) -- (L1);
\draw[monomorphism] (D2P) -- (H2);
\draw[altmonomorphism] (D2P) -- (G2);
\draw[monomorphism] (D1P) -- (G2);
\draw[altmonomorphism] (D1P) -- (H1);
\draw[morphism] (H) -- node[left]{\small $m$} (G2);
\draw[morphism] (G) -- (H1);
\draw[morphism] (D1) -- (D1P);
\draw[morphism] (D2) -- (D2P);
\draw[morphism] (M) -- (H2);
\draw[monomorphism] (K1) -- (R1);
\draw[monomorphism] (D1) -- node[below]{\small $k_1$} (H);
\draw[altmonomorphism] (D1) --  node[below]{\small $c_1$} (G);
\draw[morphism] (L1) -- node[left]{} (G);
\draw[morphism] (K1) -- (D1);
\draw[morphism] (R1) -- node[above =8pt, left = 8pt]{\small $i_1$} (H);
\draw[altmonomorphism] (K2) -- (L2);
\draw[monomorphism] (K2) -- (R2);
\draw[altmonomorphism] (D2) -- node[below]{\small $k_2$} (H);
\draw[monomorphism] (D2) -- node[below]{\small $c_2$} (M);
\draw[morphism] (L2) -- node[above=8pt, right = 8pt]{\small $i_2$} (H);
\draw[morphism] (K2) -- (D2);
\draw[morphism] (R2) -- (M);
\draw[morphism,dotted] (L2) -- node[below = 15pt, left=8pt]{\small $d_{21}$} (D1);
\draw[morphism,dotted] (R1) -- node[below = 15pt, right=8pt]{\small $d_{12}$} (D2);
\node(acL1)[outer sep=0pt,inner sep=0pt,node distance=0em,strictly right of=R1]{
\tikz[baseline,draw=black,fill=lightgray]{\filldraw (0,0) -- node[above,pos=0.6,overlay,outer sep=1ex](acL1){\small $\ac_{L_1}$} (0.3,0.12) -- (0.3,-0.12) -- (0,0);}};
\node(acL2)[outer sep=0pt,inner sep=0pt,node distance=1em,left of=L2]{\tikz[baseline,draw=black,fill=lightgray]{\filldraw (0,0) -- node[above,pos=0.6,overlay,outer sep=1ex](acL2){\small $\ac_{L_2}$} (-0.3,0.12) -- (-0.3,-0.12) -- (0,0);}};}
\end{tikzpicture}
\caption{initial parallel independent transformation pair $tp_{L_1 + L_2}$ for parallel independent AC-disregarding transformation pair $tp_G$}
\label{fig:embedding-pi}
\end{figure}

Because of uniqueness of initial transformation pairs up to isomorphism, it thus follows that for each pair of \emph{plain rules} $\tuple{p_1,p_2}$ there is a \emph{unique initial parallel independent transformation pair} $tp_{L_1 +L_2}: R_1 + L_2 \ldder_{p_1, i_1} L_1 + L_2 \dder_{p_2, i_2} L_1 + R_2$.  Note that this is different from the situation for conflicts for plain rules, where initial transformation pairs may differ from conflict to conflict. Consequently, in general for a pair of rules we can have different initial conflicts, but there exists always a unique initial parallel independent transformation pair.

\begin{definition}[initial parallel independent transformation pair]\label{def:core-parallel}
A pair of parallel independent plain transformations $tp: H_1\ldder_{p_1,m_1}G\dder_{p_2,m_2} H_2$ is an \emph{initial parallel independent transformation pair} if it is isomorphic to the transformation pair $tp_{L_1 +L_2}: R_1 + L_2 \ldder_{p_1, i_1} L_1 + L_2 \dder_{p_2, i_2} L_1 + R_2$.
\end{definition}

The one-element set consisting of the initial parallel independent transformation pair for a given pair of rules is \emph{complete w.r.t. parallel independence}. 

\begin{theorem}[completeness of initial parallel independent transformation pairs]
\label{thm:completeness-core-parallel}
The set consisting of the initial parallel independent transformation pair $tp_{L_1 +L_2}: R_1 + L_2 \ldder_{p_1, i_1} L_1 + L_2 \dder_{p_2, i_2} L_1 + R_2$ for a pair of plain rules $\tuple{p_1,p_2}$ is \emph{complete} w.r.t. parallel independence. 
\end{theorem}
\begin{proof}
This follows directly from \autoref{lem:initial-parallel-independent} and \autoref{def:core-parallel}. \qed
\end{proof}

\section{Initial Conflicts}\label{sec:cond-conflict}


We start with showing why it is not possible to straightforwardly generalize the idea of initial conflicts from plain rules to rules with ACs. On the one hand, \emph{conflict inheritance} does not hold any more such that not each transformation pair that can be embedded into a conflicting one is conflicting again, which was the basis for being able to show completeness of initial conflicts for plain rules. Actually, the reverse of inheritance, what we call conflict co-inheritance, does not hold either, i.e., not each transformation pair that can embed a conflicting one is conflicting again (cf.~\autoref{subsec:conflict-inheritance}). 
Moreover, it is \emph{impossible} in general to find a \emph{finite and complete subset of finite conflicts} for rules with ACs (cf.~\autoref{subsec:infinity}) as illustrated for the category of {\bf Graphs}. Finiteness is a basic prerequisite however to be able to practically compute a complete (i.e. representative) set of conflicts statically. This motivates again the need for having \emph{symbolic transformation pairs} as introduced in \autoref{def:symbolic-pair}, allowing us to define \emph{initial conflicts} (cf.~\autoref{subsec:def-initial-cond-conflict}) as a  set of specific symbolic transformation pairs, being complete w.r.t. parallel dependence indeed (as shown in \autoref{subsec:completeness-confluence}). This set as well as its elements are also finite, for example, in the case of graphs (and provided that the rules are finite).  

\subsection{Conflict Inheritance}\label{subsec:conflict-inheritance}
Conflicts are in general not inherited (as opposed to the case of plain rules~\cite{LambersBO0T18}) such that not each (initial) transformation pair that can be embedded into a conflicting one will be conflicting again. This may happen in particular for AC conflicts. Use-delete (resp. delete-use) conflicts for rules with ACs are still inherited. 

\begin{lemma}[Use-delete (delete-use) conflict inheritance]\label{lem:inheritance}
Given a pair of direct transformations $tp$ in use-delete (resp. delete-use) conflict and another pair of direct transformations $tp'$ that can be embedded into $tp$ via extension morphism $f$  and corresponding extension diagrams, then $tp'$ is also in use-delete (resp. delete-use) conflict.  
\end{lemma}
\begin{proof}
The proof is completely analogous to the one for the conflict inheritance lemma for plain rules in~\cite{LambersBO0T18} and use-delete (resp. delete-use) conflicts.  \qed 
\end{proof}

\begin{example}[Neither inheritance nor co-inheritance for AC conflicts]\label{example:inheritance}
Consider rules 
$p_1: \node{} \leftarrow \node{} \rightarrow \node{}\outgoing\node{}$
 (with AC $\true$), producing an outgoing edge with a node, and 
$p_2: \node{} \leftarrow \node{} \rightarrow \node{} \ \node{}$
 with NAC 
$\neg \exists n: \node{} \rightarrow \node{} \ \node{} \ \node{}$ , producing a node only if two other nodes do not exist already. Consider graph 
$G= \node{} \ \node{}$, holding two nodes. Applying both rules to $G$ (with the matches sharing one node in $G$) we obtain a produce-AC conflict since the first rule creates a third node, forbidden by the second rule. Now both rules can be applied similarly to the shared node in the subgraph 
$G' = \node{}$ of $G$ obtaining parallel independent transformations, illustrating that AC-conflicts are not inherited.

Assume that $p_2$ would have the more complex AC 
$(\neg \exists n: \node{} \rightarrow \node{} \ \node{} \ \node{}) \vee (\exists p: \node{} \rightarrow \node{} \ \node{} \ \node{} \ \node{})$, then the transformation pair arising from their application to $G$ sharing one node with their matches is still produce-AC conflicting. Now the application of both rules to the extended graph 
$G'' = \node{} \ \node{} \ \node{} \ \node{}$ (sharing with the extended matches the same node as in $G$) would satisfy the AC and would be moreover parallel independent, illustrating that AC-conflicts are not co-inherited. 
\end{example}

\subsection{Complete Subset of Conflicts}\label{subsec:infinity}
We show that in $\M$-adhesive categories it is in general impossible to find a finite and complete subset of finite conflicts for rules with ACs as illustrated for the category {\bf Graphs} (under the assumption\footnote{Without this assumption even in the case of plain rules the set of critical pairs would already be infinite.} that graph transformation rules are finite). 

\begin{theorem}\label{thm:infinite-set-complete}
Given finite rules $\prule_1 = \tuple{p_1,\ac_{L_1}}$ and $\prule_2 = \tuple{p_2,\ac_{L_2}}$ for the $\M$-adhesive category {\bf Graphs}, in general, there is no finite set of finite  transformation pairs $\SC$ for $\prule_1$ and $\prule_2$ that is complete w.r.t. parallel dependence.
\end{theorem}
\begin{proof}[idea]
The idea of the proof is that if such a finite set $\S$ always exists, we can derive that each first-order formula is  equivalent to a finite disjunction of atomic formulas, which is a contradiction. To show this, we define rules
$\prule_1 = \langle \emptyset \leftarrow \emptyset \rightarrow \emptyset, c \rangle$ and 
$\prule_2 = \langle\emptyset \leftarrow \emptyset \rightarrow 1_N, \true \rangle$, with $1_N$ the graph consisting just of an isolated node and $c$ some arbitrary property (expressible using ACs over the empty graph) about graphs without isolated nodes. Assuming that 
$\SC_0 = \{G \Longleftarrow_{\prule_1} G \Longrightarrow_{\prule_2} G\oplus 1_N \mid G\models c \}$ is  the set of transformation pairs associated to these rules, then we can show that $c$ is equivalent to a finite disjunction of existential atoms of the form $ \exists(\emptyset \rightarrow H, \true)$
\qed
\end{proof}

\subsection{Initial Conflicts}\label{subsec:def-initial-cond-conflict}
We generalize the notion of \emph{initial conflicts} for plain rules to rules with ACs.  In particular, we introduce them as special symbolic transformation pairs.  They are \emph{conflict-inducing} meaning that there needs to exist an unfolding of the symbolic transformation pair into a concrete conflicting transformation pair. Moreover, their AC-disregarding transformation pair needs to be an initial conflict or initial parallel independent transformation pair.  We also show formally the \emph{relationship between initial conflicts and critical pairs} as reintroduced in \autoref{subsec:critical-pairs}. In particular, we demonstrate that initial conflicts represent a proper subset of critical pairs again. 

\begin{definition}[unfolding of symbolic transformation pair]
Given a symbolic transformation pair $stp_K: \tuple{tp_K,\ac_K,\ac^*_K}$ for rule pair $\tuple{\prule_1,\prule_2}$, then its unfolding $\mathcal{U}(stp_K)$ consists of all transformation pairs $H_1\ldder_{\prule_1,m_1}G\dder_{\prule_2,m_2} H_2$ representing the lower row of the extension diagrams via some extension morphism $m:K\rightarrow G$ as shown in \autoref{fig:completeness} (with AC-disregarding transformation pair $tp_K$ in the upper row).  
\end{definition}

\begin{remark}[non-empty unfolding]
Note that the unfolding of a symbolic transformation pair is not empty if there exists an extension morphism $m:K\rightarrow G$  satisfying the gluing conditions as well as $\ac_K$ for the  derived spans (as can be followed directly from the Embedding Theorem~\cite{EhrigHLOG10,EhrigGHLO12} for rules with ACs, since $m$ would be boundary as well as AC-consistent).  
\end{remark}

\begin{definition}[conflict-inducing symbolic transformation pair]\label{def:symbolic-pair-conflict} Given rules $\prule_1=\tuple{p_1,\ac_{L_1}}$ and $\prule_2=\tuple{p_2,\ac_{L_2}}$, a symbolic transformation pair $stp_K: \tuple{tp_K,\ac_K,\ac^*_K}$ for $\tuple{\prule_1,\prule_2}$ is \emph{conflict-inducing} if there exists a pair of conflicting transformations in its unfolding $\mathcal{U}(stp_K)$. 
\end{definition}

\begin{remark}[conflict-inducing \& unfolding]\label{rem:unfolding} The unfolding of a conflict-inducing symbolic transformation pair may contain parallel independent transformations. Consider rules $\prule_1=\tuple{p_1,\true}$ and $\prule_2=\tuple{p_2,\neg \exists n}$ from \autoref{example:inheritance} and symbolic transformation pair $stp': \tuple{tp_{G'},\ac_{G'},\ac^*_{G'}}$, with $tp_{G'}$ the AC-disregarding transformation pair arising from applying rules $p_1$ and $p_2$ to $G' = \node{}$, $\ac_{G'} = \NE n': \node{} \rightarrow \node{} \ \node{} \ \node{}$, and $\ac^*_{G'} = (\PE p': \node{} \rightarrow \node{} \ \node{}) \vee (\PE p'': \node{} \rightarrow \node{} \ \node{} \ \node{})$.  
Then $stp'$ is a conflict-inducing symbolic transformation pair, since its unfolding includes the parallel dependent transformation pair $tp_{G}$ arising from applying the rules $\prule_1$ and $\prule_2$ to $G = \node{} \ \node{}$. The extension morphism $m: G' \rightarrow G$ fulfills $\ac_{G'}$ and $\ac^*_{G'}$ indeed. However, the transformation pair $tp_{G'}$ satisfies all ACs, belongs to the unfolding $\mathcal{U}(stp)$ accordingly, but is parallel independent (as described in \autoref{example:inheritance} and derivable from the fact that $\ac^*_{G'}$ is not fulfilled for the extension morphism $id_{G'}$). 
\end{remark}

%

An \emph{initial conflict} is a conflict-inducing symbolic transformation pair with its AC-disregarding transformation pair being initial. Note that we say that an AC-disregarding transformation pair is initial if it is initial as plain transformation pair (cf.~\autoref{fig:embedding-pi}). Remember that each symbolic transformation pair is uniquely determined by its underlying AC-disregarding transformation pair. This means that the set of initial conflicts basically consists of a filtered set of plain initial conflicts (those that are conflict-inducing as symbolic transformation pair) together with the initial parallel independent transformation pair (in case it is conflict-inducing as symbolic transformation pair). 

\begin{definition}[initial conflict]\label{def:initial-conflict} Consider an $\mathcal{M}$-adhesive system with initial transformation pairs for conflicts along plain rules. An \emph{initial conflict} for rules $\prule_1=\tuple{p_1,\ac_{L_1}}$ and $\prule_2=\tuple{p_2,\ac_{L_2}}$ is a conflict-inducing symbolic transformation pair $stp_K: \tuple{tp_K,\ac_K,\ac^*_K}$ with the AC-disregarding transformation pair $tp_K$ being initial, i.e. either $tp_K$ is an initial conflict for rules $p_1$ and $p_2$ (in this case $stp_K$ is called a \emph{use-delete/delete-use initial conflict}) or it is the initial parallel independent transformation pair $tp_{L_1+L_2}$ for rules $p_1$ and $p_2$ (in this case $stp_{K} = stp_{L_1+L_2} = \tuple{tp_{L_1+L_2},\ac_{L_1+L_2},\ac^*_{L_1+L_2}}$ is called the \emph{AC initial conflict}).  
\end{definition}


Note that as explained in \autoref{rem:unfolding} the unfolding of a conflict-inducing symbolic transformation pair (and in particular of an AC initial conflict) may entail apart from (at least one) conflicting transformation pair(s) also parallel independent transformation pairs.  All conflicts in the unfolding of an AC initial conflict are AC conflicts, and never use-delete/delete-use conflicts (because otherwise we would get a contradiction using \autoref{lem:inheritance}). 

\begin{example}[initial conflict]\label{example:ic}
Consider again the rules from \autoref{example:inheritance}.  Applying both rules to $L_1 +L_2 = \node{} \ \node{}$ (with disjoint matches) we obtain the AC initial conflict $stp_{K} = stp_{L_1+L_2} = \tuple{tp_{L_1+L_2},\ac_{L_1+L_2},\ac^*_{L_1+L_2}}$. Thereby $\ac_{L_1 + L_2}$ is equivalent to $\neg \exists (\node{1} \ \node{2} \rightarrow \node{1} \ \node{2} \ \node{}) \wedge \neg \exists (\node{1} \ \node{2} \rightarrow \node{1,2} \ \ \ \node{} \ \node{})$, expressing that when during extension both nodes are merged, no two additional nodes, otherwise not one additional node should be given. Moreover, $\ac^*_{L_1+L_2}$ is equivalent to $\exists (\node{1} \ \node{2} \rightarrow \node{1,2} \ \ \ \node{}) \vee \exists(\node{1} \ \node{2} \rightarrow \node{1} \ \node{2})$, expressing that either both nodes are not merged during extension, otherwise one additional node should be present for a conflict to arise.  Both transformation pairs (the conflicting one from $G = \node{} \ \node{}$ as well as the parallel independent one from its subgraph $G' = \node{}$, sharing the merged node in their matches) described in \autoref{example:inheritance} belong to its unfolding.   
\end{example}

Each initial conflict is in particular also a critical pair.

\begin{theorem}[initial conflict is critical pair]\label{thm:ic-is-cp} Consider an $\mathcal{M}$-adhesive system with initial transformation pairs for conflicts along plain rules. Each initial conflict $stp_K: \tuple{tp_K,\ac_K,\ac^*_K}$ is a critical pair. 
\end{theorem}
\begin{proof}[idea]
The proof is quite straightforward. In particular, it is routine to show that the matches defining $stp_K: \tuple{tp_K,\ac_K,\ac^*_K}$ are in $\mathcal{E'}$ and that there exists a morphism $m: K \rightarrow G \in \mathcal{M}$ satisfying all further conditions from \autoref{def:critical} guaranteeing that $stp_K$ is conflict-inducing.\qed 
\end{proof}

The reverse direction of \autoref{thm:ic-is-cp} does not hold, i.e. in the case of rules with ACs initial conflicts represent also  a \emph{proper subset} of the set of critical pairs. This proper subset relation holds already in the case of plain rules. Therefore, in the case that $\ac_{L_1}$ and $\ac_{L_2}$ are true, a critical pair $stp_K:\tuple{tp_K,\ac_K,\ac^*_K}$ has $\ac_K$ true and $\ac^*_{K}$ true, since either $\ac^*_{K,d_{12}}$ or $\ac^*_{K,d_{21}}$ needs to be false with $tp_K$ a use-delete/delete-use conflict.  Since such a $tp_K$ is in particular a critical pair for plain rules, this means that we would have as many critical pairs that are no initial conflicts as for the case with plain rules. More generally, critical pairs $stp_K:\tuple{tp_K,\ac_K,\ac^*_K}$ where $tp_K$ represents a use-delete/delete-use conflict (but is not initial yet) are represented by the initial conflict $stp_I:\tuple{tp_I,\ac_I,\ac^*_I}$ with $tp_I$ the unique initial conflict for $tp_K$ as plain transformation pair. Moreover, critical pairs $stp_K:\tuple{tp_K,\ac_K,\ac^*_K}$  where $tp_K$ is parallel independent as plain transformation pair are represented by one initial conflict $stp_{L_1 +L_2}:\tuple{tp_{L_1 +L_2},\ac_{L_1 +L_2},\ac^*_{L_1 +L_2}}$  with $tp_{L_1 +L_2}$ the initial parallel independent transformation pair. 

\begin{example}[initial conflicts: proper subset of critical pairs]\label{example:subset}
Consider again the rules from \autoref{example:inheritance} and their application to $G' = \node{}$. The symbolic transformation pair $stp_{G'}: \tuple{tp_{G'},\ac_{G'},\ac^*_{G'}}$ is a critical pair, but not an initial conflict. In particular, this critical pair is represented by the unique AC initial conflict $stp_{L_1 +L_2}:\tuple{tp_{L_1 +L_2},\ac_{L_1 +L_2},\ac^*_{L_1 +L_2}}$ (which is also a critical pair).   
\end{example}

\subsection{Completeness}
\label{subsec:completeness-confluence}
We show that initial conflicts are complete (not $\mathcal{M}$-complete as in the case of critical pairs, cf.~\autoref{thm:completeness}) w.r.t. parallel dependence as symbolic transformation pairs. 

\begin{theorem}[completeness of initial conflicts]\label{thm:completeness-ic-acs} Consider an $\mathcal{M}$-adhesive system with initial transformation pairs for conflicts along plain rules. The set of \emph{initial conflicts} for a pair of rules $\tuple{\prule_1,\prule_2}$ is \emph{complete} w.r.t. parallel dependence. 
\end{theorem}
\begin{proof}[idea]
Let $tp_G$ be a parallel dependent pair of transformations for $\tuple{\prule_1,\prule_2}$. The theorem is a direct consequence of: a) if $tp_G$ is a use-delete/delete-use conflict, according to Thm. \ref{thm:completeness-core}, there is an initial conflict $tp_K$ that can be embedded into $tp_G$; and b) if $tp_G$ is an AC-conflict, according to Lemma \ref{lem:initial-parallel-independent}, $tp_{L_1+L_2}$ can be embedded into $tp_G$, where $L_1, L_2$ are the left hand sides of the rules $\prule_1,\prule_2$, respectively. \qed
\end{proof}

\begin{remark}[uniqueness of initial conflicts]
It holds again that for each conflict a \emph{unique} (up-to-isomorphism) initial conflict exists representing it, since this property is inherited from the one for plain rules~\cite{LambersBO0T18} and the fact that the initial parallel independent pair of transformations is unique w.r.t. a given rule pair.  
\end{remark}


Initial conflicts are also minimally complete, i.e. we are able to generalize the corresponding result for plain rules (cf.~\autoref{cor:minimal}) to rules with ACs.

\begin{corollary}[minimally complete]\label{cor:minimal-AC}
Consider an $\mathcal{M}$-adhesive system with initial transformation pairs for conflicts via plain rules. The set of initial conflicts $\mathcal{S}$ (up-to-isomorphism) for rules $\tuple{\prule_1,\prule_2}$ is \emph{minimally complete} w.r.t. parallel dependence, i.e. there does not exist any smaller set $\mathcal{S'}$ of symbolic transformation pairs for $\tuple{\prule_1,\prule_2}$ being complete w.r.t. parallel dependence. 
\end{corollary}


The Local Confluence Theorem~\cite{EhrigHLOG10,EhrigGHLO12} for rules with ACs\footnote{On top of strict confluence as in the case of plain rules, also so-called AC-compatibility is required.} still holds in case we substitute the set of critical pairs by initial conflicts, and moreover requiring initial pushouts.  The proof runs completely analogously. The only difference is that for this proof, we need initial pushouts over general morphisms whereas in the proof in~\cite{EhrigHLOG10,EhrigGHLO12} initial pushouts over $\mathcal{M}$-morphisms are sufficient. 

\section{Unfoldings of Initial Conflicts}\label{sec:unfolding}
We show a \emph{sufficient condition} for being able to unfold  initial conflicts into a \emph{complete set of conflicts} that is \emph{finite} if the set of initial conflicts is finite (cf.~\autoref{subsec:unfolding}). We demonstrate moreover that this sufficient condition is fulfilled for the special case of having merely \emph{NACs} as rule application conditions (cf.~\autoref{subsec:nac}). Finally, we show that in this case we obtain in particular specific critical pairs for rules with negative application conditions (NACs) as introduced in \cite{LambersEO06} again. In this sense we show explicitly that initial conflicts as introduced in this paper represent a conservative extension of the critical pair theory for rules with NACs.

\subsection{Finite and Complete Unfolding}\label{subsec:unfolding}
We introduce so-called \emph{regular initial conflicts} leading  to $\mathcal{M}$-complete subsets of conflicts by unfolding them in some particular way (cf. \emph{disjunctive unfolding} in \autoref{def:finite-unfolding}). The idea is that the extension and conflict-inducing AC ($\ac_K$ and $\ac^*_K$, respectively) of such a regular initial conflict $stp_K:\tuple{tp_K,\ac_K,\ac^*_K}$ have a specific form that is amenable to finding $\mathcal{M}$-complete unfoldings.  We expect the condition $\ac_K \wedge \ac^*_K$ to consist of a \emph{disjunction of positive literals}  (conditions of the form $\exists (a_i:K \rightarrow C_i, c_i)$) with a so-called \emph{negative remainder} (i.e. a condition $c_i = \wedge_{j \in J} \neg \exists (b_j: C_i \rightarrow C_j,d_j)$).  Intuitively, this means that there is a finite number of possibilities to unfold the symbolic conflict into a concrete conflict by adding some specific positive context (expressed by the morphism $a_i$). The negative remainder $c_i$ ensures that by adding this positive context to the context $K$ of the symbolic transformation pair within the initial conflict, we indeed find a concrete conflict when not extending further at all. Moreover, it expresses under which condition the corresponding concrete representative conflict leads to further conflicts by extension. Finally, the subsets of $\M$-complete conflicts built using the disjunctive unfolding can shown to be \emph{finite} if the set of initial conflicts it is derived from is finite.
  
\begin{definition}[regular initial conflict, disjunctive unfolding]\label{def:finite-unfolding} Consider an $\mathcal{M}$-adhesive system with initial transformation pairs for conflicts along plain rules. 
Given an initial conflict $stp_K:\tuple{tp_K,\ac_K,\ac^*_K}$ for  rules $\tuple{\prule_1,\prule_2}$, then we say that it is \emph{regular} if $\ac_K \wedge \ac^*_K$ is equivalent to a condition $\vee_{i \in I} \exists (a_i:K \rightarrow C_i, c_i)$ with $c_i = \wedge_{j \in J} \neg \exists (b_j: C_i \rightarrow C_j,d_j)$ a condition on $C_i$, $b_j$ non-isomorphic and $I$ some non-empty index set. Given a regular initial conflict $stp_K:\tuple{tp_K,\ac_K,\ac^*_K}$, then  $\mathcal{U^D}(stp_K) = \cup_{i \in I}\{tp_{C_i}: D_{1,i}\ldder_{\prule_1,a_i \circ o_1} C_i \dder_{\prule_2,a_i \circ o_2} D_{2,i}\}$ is the \emph{disjunctive unfolding} of $stp_K$.
\end{definition}

\begin{remark}[disjunctive unfolding]
The disjunctive unfolding of a regular conflict is non-empty, but might consist of less elements than literals in the disjunction $\vee_{i \in I} \exists (a_i:K \rightarrow C_i, c_i)$.  It might be the case that some of the morphisms $a_i$ do not satisfy the gluing condition of the derived spans. If this is the case, then also every extension morphism starting from there will not satisfy the gluing condition such that we can safely ignore these cases from the disjunctive unfolding.
\end{remark}

\begin{theorem}[finite and complete unfolding]\label{thm:complete-unfolding} Consider an $\mathcal{M}$-adhesive system with initial transformation pairs for conflicts along plain rules. Given a rule pair $\tuple{\prule_1,\prule_2}$ with set $\mathcal{S}$ of initial conflicts such that each initial conflict $stp$ in $S$ is regular, then $\cup_{stp \in \mathcal{S}}\mathcal{U^D}(stp)$ is $\mathcal{M}$-complete w.r.t. parallel dependence. Moreover, $\cup_{stp \in \mathcal{S}}~\mathcal{U^D}(stp)$ is finite if $\mathcal{S}$ is finite. 
\end{theorem}
\begin{proof}[idea]
From \autoref{thm:completeness-ic-acs} we know there exists some $stp_K:\tuple{tp_K,\ac_K,\ac^*_K}$ from $S$ that can be embedded into $tp_G$ via some extension morphism $m$ and, since $m \models \ac_K \wedge \ac^*_K$ and $stp_K$ is regular, $m \models  \exists (a_i:K \rightarrow C_i, c_i)$ for some $i\in I$. Then the key issue is to show that $tp_{C_i}: D_{1,i}\ldder_{\prule_1,a_i \circ o_1} C_i \dder_{\prule_2,a_i \circ o_2} D_{2,i}$ can be embedded into $tp_G$ and that $tp_{C_i}$ is conflicting. Finally, finiteness of $\cup_{stp \in \mathcal{S}}~\mathcal{U^D}(stp)$ is a consequence of the finiteness of each $\mathcal{U^D}(stp)$.
\qed
\end{proof}

It is possible to automatically check if some initial conflict is regular by using dedicated automated reasoning~\cite{LOicgt14} as well as symbolic model generation for ACs~\cite{SchneiderLO17} as follows. The reasoning mechanism~ \cite{LOicgt14} is shown to be refutationally complete ensuring that if the condition $\ac_K \wedge \ac^*_K$  of some initial conflict is unsatisfiable, this will be detected eventually.  Moreover, the related symbolic model generation mechanism~\cite{SchneiderLO17} is able to automatically transform each condition $\ac_K \wedge \ac^*_K$ into some  disjunction $\vee_{i \in I} \exists (a_i:K \rightarrow C_i, c_i)$ with $c_i$ a negative remainder if such an equivalence holds.

\subsection{Unfolding for Rules with NACs}\label{subsec:nac}
We show that in the case of having rules with NACs\footnote{A rule with NACs consists of a plain rule with a conjunction of NACs as application condition, which is the most common way of using NACs since their introduction in \cite{HabelHT96}.}, initial conflicts are regular.  This means that in this special case there exists a complete subset of conflicts that is e.g. in the case of graphs (and assuming finite rules) also finite.  This conforms to the findings in~\cite{LambersEO06,Lambers2009}, where an $\mathcal{M}$-complete set of critical pairs -- as specific subset of conflicts -- for graph transformation rules with NACs was introduced~\cite{LambersEO06} (and generalized to $\mathcal{M}$-adhesive transformation systems~\cite{Lambers2009}).

\begin{theorem}[regular initial conflicts for rules with NACs]\label{thm:monotone-NACs} Consider an $\mathcal{M}$-adhesive system with initial transformation pairs for conflicts along plain rules. Given some initial conflict $stp_K:\tuple{tp_K,\ac_K,\ac^*_K}$ for a pair of rules $\tuple{\prule_1,\prule_2}$ with $\ac_{L_i} = \wedge_{j\in J} \neg \exists{n_j : L_i \rightarrow N_j}$ for $i=1,2$ and $J$ some finite index set, then it is \emph{regular}. In particular, $\ac_K \wedge \ac^*_K$ is equivalent to a condition $\vee_{i \in I} \exists (a_i:K \rightarrow C_i, c_i)$ with $c_i = \wedge_{q \in Q} \neg \exists n_q$ a condition on $C_i$ and $I$ some non-empty index set.
\end{theorem} 
\begin{proof}[idea]
This follows from \autoref{def:symbolic-pair} and the constructions~\cite{EhrigGHLO14} related to \autoref{lem:left} and \autoref{lem:shift}.\qed
\end{proof}

The negative remainder $c_i$ of each literal in $\vee_{i \in I} \exists (a_i:K \rightarrow C_i, c_i)$ of a regular initial conflict for rules with NACs thus consists of a set of NACs. Intuitively this means that we obtain for each initial conflict an $\M$-complete subset of concrete conflicts by adding the context described by $a_i$. As long as no NAC from $c_i$ is violated we can extend such a concrete conflict to further ones.  

\begin{corollary}[complete unfolding: rules with NACs]\label{cor:monotone-NACs} Consider an $\mathcal{M}$-adhesive system, with initial transformation pairs for conflicts along plain rules. Given a rule pair $\tuple{\prule_1,\prule_2}$ with $\ac_{L_i} = \wedge_{j\in J} \neg \exists{n_j : L_i \rightarrow N_j}$ for $i=1,2$, then $\cup_{stp \in \mathcal{S}}\mathcal{U^D}(stp)$ is $\mathcal{M}$-complete w.r.t. parallel dependence.  
\end{corollary} 
\begin{proof}
This follows directly from \autoref{thm:complete-unfolding} and \autoref{thm:monotone-NACs}.
\qed
\end{proof}

We show moreover that the initial conflict definition is a \emph{conservative extension} of the critical pair definition for rules with NACs as given in~~\cite{LambersEO06,Lambers2009}. In particular, we show that each conflict in the disjunctive unfolding of an initial conflict as chosen in the proof of \autoref{thm:monotone-NACs} is in particular a critical pair for rules with NACs. Note that a critical pair for rules with NACs is a conflicting pair of transformations such that (1) its plain transformations have jointly surjective matches and are use-delete/delete-use conflicting, or (2) the transformations are AC conflicting (and possibly also use-delete/delete-use conflicting) in such a way that one of the rules produces elements responsible for violating one of the NACs not violated yet before rule application without considering additional context not stemming already from one of the rules or the violated NAC (i.e. technically the morphism violating the NAC and the corresponding co-match need to be jointly surjective).  

\begin{theorem}[conservative unfolding]\label{thm:conservative-NACs} Consider an $\mathcal{M}$-adhesive system with initial transformation pairs for conflicts along plain rules. Given some initial conflict $stp_K:\tuple{tp_K,\ac_K,\ac^*_K}$ for a pair of rules $\tuple{\prule_1,\prule_2}$ with $\ac_{L_i} = \wedge_{j\in J} \neg \exists{n_j : L_i \rightarrow N_j}$ for $i=1,2$ and $J$ some finite index set, then each conflict as chosen in the proof of \autoref{thm:monotone-NACs} in $\mathcal{U^D}(stp)$ is in particular a critical pair for $\tuple{\prule_1,\prule_2}$ as given in~\cite{LambersEO06,Lambers2009}.  
\end{theorem}

\begin{example}[conservative unfolding]\label{example:conservative}
Consider again the rules from \autoref{example:inheritance} (having only NACs as ACs) and their application to the graph $G = \node{} \ \node{}$. The corresponding  transformation pair $tp_{G}$ is a critical pair for rules with NACs as given in~\cite{LambersEO06,Lambers2009}. This is because it is in particular a conflicting pair of transformations, and the morphism violating the NAC (since finding the three nodes) and therefore causing the conflict after applying the first rule to $G = \node{} \ \node{}$ obtaining some graph $H_1 = \node{}\outgoing\node{} \ \node{}$ is jointly surjective together with the corresponding co-match. As argued already in \autoref{example:ic} this critical pair for rules with NACs belongs to the unfolding (and in particular to the disjunctive unfolding) of the unique AC initial conflict $stp_{L_1 +L_2}:\tuple{tp_{L_1 +L_2},\ac_{L_1 +L_2},\ac^*_{L_1 +L_2}}$.   
\end{example}

\section{Conclusion and Outlook}\label{sec:conclusion}

In this paper we have \emph{generalized the theory of initial conflicts} (from plain rules, i.e. rules without application conditions) to \emph{rules with application conditions} (ACs) in the framework of $\M$-adhesive transformation systems.  We build on the notion of symbolic transformation pairs, since it turns out that it is not possible to find a complete subset of concrete conflicting transformation pairs in the case of rules with ACs. We have shown that that initial conflicts are (minimally) complete w.r.t. parallel dependence as symbolic transformation pairs. Moreover, initial conflicts represent (analogous to the case of plain rules) proper subsets of critical pairs in the sense that for each critical pair (or also for each conflict), there exists a unique initial conflict representing it.  We concluded the paper by showing sufficient conditions for finding unfoldings of initial conflicts that lead to (finite and) complete subsets of conflicts (as in the case of rules with NACs).  Thereby we have shown that initial conflicts for rules with ACs represent a conservative extension of the critical pair theory for rules with NACs.

As future work we aim at finding \emph{further interesting classes} allowing finite and (minimally) complete unfoldings into \emph{subsets of conflicts}. This will serve as a guideline to be able to \emph{develop and implement efficient conflict detection} techniques for rules with (specific) ACs, which has been an open challenge until today.  
We are moreover planning to develop (semi-)automated detection of unfoldings of initial conflicts of  rules with arbitrary ACs using dedicated automated reasoning and model finding for graph conditions~\cite{Pennemann2009,LOicgt14,SchneiderLO17}. 
It would  be interesting to investigate in which \emph{use cases} initial conflicts (or critical pairs) are useful already as symbolic transformation pairs, and in which use cases we rather need to consider unfoldings indeed. This is in line with the research on multi-granular conflict detection~\cite{BornL0T17,Lambers0TBH18,LambersBKST19} investigating different levels of granularity that can be interesting from the point of view of applying conflict detection to different use cases. 
Finally, we plan to investigate conflict detection in the light of initial conflict theory for \emph{attributed graph transformation}~\cite{EhrigEPT06,HristakievP16,KulcsarDLVS15}, and in particular the case of rules with so-called attribute conditions more specifically. It would also be interesting to further investigate initial conflicts for transformation rules (with ACs) not following the DPO approach. For example, one may consider the single-pushout (SPO) approach introduced in \cite{Lowe93}, which is a generalization of the DPO framework where only one morphism defines the rule, which may be partial to allow deletion. In \cite{HabelHT96}, SPO rules with negative application conditions are considered and the Local Confluence and Parallelism Theorems are shown. As far as we know, a theory on SPO rules with nested application conditions is missing. Moreover, the implications of initial conflict theory for the case of \emph{graphs with inheritance}~\cite{GolasLEO12} or \emph{rule amalgamation}~\cite{TaentzerG15,BornT16} need to be further investigated. 

\textbf{Acknowledgement.} We thank Jens Kosiol for pointing out that the set of initial conflicts for plain rules is not only complete, but also minimally complete. We were able to transfer this result to rules with ACs in this paper. Many thanks also to the reviewers for their detailed and constructive comments helping to finalize the paper. 

%
%
\bibliographystyle{splncs04}
\bibliography{literature}
\newpage

\appendix
\section{Proofs}\label{sec:proofs}

{\noindent \bf Proof of \autoref{cor:minimal}}
\begin{proof}
Assume that $\mathcal{S'}$ exists. By completeness w.r.t. parallel dependence of $\mathcal{S'}$, and since the cardinality of $\mathcal{S'}$ is smaller than that of $\mathcal{S}$, there exists at least one $tp: H_1\ldder_{p_1,m_1}G\dder_{p_2,m_2} H_2 \in \mathcal{S'}$ such that $tp$ can be embedded into two non-isomorphic initial conflicts $tp': H'_1\ldder_{p_1,m'_1}G'\dder_{p_2,m_2} H'_2$ and $tp'': H''_1\ldder_{p_1,m''_1}G''\dder_{p_2,m''_2} H''_2$ in $\mathcal{S}$ via extension morphisms $m':G\rightarrow G'$ and $m'':G\rightarrow G''$, respectively. Since $tp'$ as well as $tp''$ are initial conflicts, it follows that $tp'$ and $tp''$ are in particular initial transformation pairs for $tp'$ and $tp''$, respectively.  Consequently $tp'$ and $tp''$ are initial w.r.t. $tp$, otherwise this would lead to a contradiction with being initial to $tp'$ and $tp''$. 
Because of uniqueness of initial transformation pairs it would follow that $tp'$ and $tp''$ are isomorphic, which is a contradiction.
\qed
\end{proof}


To prove \autoref{lem:initial-parallel-independent} we will use the following lemma:

\begin{lemma}[extensions of coproduct transformation pair]
\label{lem:coproduct extensions} 
Given rules $p1: L_1 \leftarrow I_1 \rightarrow R_1$ and $p2: L_2 \leftarrow I_2 \rightarrow R_2$ and transformation pairs $tp: H_1\ldder_{p_1,m_1}G\dder_{p_2,m_2} H_2$ and $tp_{L_1 +L_2}: R_1 + L_2 \ldder_{p_1, i_1} L_1 + L_2 \dder_{p_2, i_2} L_1 + R_2$, where $tp$ is parallel independent, we have that the coproduct mediating morphism $m: L_1 + L_2 \rightarrow G$ defines the extension diagram:

$$
    \xymatrix{
          R_1 + L_2 \ar@{->}[d]_{}  
          & L_1 + L_2  \ar@{=>}[l] \ar@{=>}[r] \ar@{->}[d]^{m}  
          & L_1 + R_2 \ar@{->}[d]          \\
          H_1 
          & G  \ar@{=>}[l] \ar@{=>}[r] 
          & H_2  \\
  }
  $$

\end{lemma}

\begin{proof}
Let $tp$ be:

\[\tikz[node distance=4em,shape=rectangle,outer sep=1pt,inner sep=2pt,label distance=-1.25em]{
\node(H){$G$};
\node(D1)[node distance=7em,left of=H]{$D_1$};
\node(G)[left of=D1]{$H_1$};
\node(L1)[above of=G]{$R_1$};
\node(K1)[right of=L1]{$I_1$};
\node(R1)[right of=K1]{$L_1$};
\node(D2)[node distance=7em,right of=H]{$D_2$};
\node(M)[right of=D2]{$H_2$};
\node(R2)[above of=M]{$R_2$};
\node(K2)[left of=R2]{$I_2$};
\node(L2)[left of=K2]{$L_2$};
\draw[altmonomorphism] (K1) -- (L1);
\draw[monomorphism] (K1) -- (R1);
\draw[monomorphism] (D1) -- node[below]{\small $k_1$} (H);
\draw[altmonomorphism] (D1) --  node[below]{\small $c_1$} (G);
\draw[morphism] (L1) -- node[left]{} (G);
\draw[morphism] (K1) -- (D1);
\draw[morphism] (R1) -- node[below=5pt]{\small $m_1$} (H);
\draw[altmonomorphism] (K2) -- (L2);
\draw[monomorphism] (K2) -- (R2);
\draw[altmonomorphism] (D2) -- node[below]{\small $k_2$} (H);
\draw[monomorphism] (D2) -- node[below]{\small $c_2$} (M);
\draw[morphism] (L2) -- node[below=5pt]{\small $m_2$} (H);
\draw[morphism] (K2) -- (D2);
\draw[morphism] (R2) -- (M);
\draw[morphism,dotted] (L2) -- node[left=15pt]{\small $d_{21}$} (D1);
\draw[morphism,dotted] (R1) -- node[right=15pt]{\small $d_{12}$} (D2);}\]
Let us prove that if $m: L_1 + L_2 \rightarrow G$ is the mediating morphism for the coproduct, satisfying that $m\circ i_1 = m_1$ and $m\circ i_2 = m_2$, where 
$i_1: L_1 \rightarrow L_1 + L_2$ and $i_2: L_2 \rightarrow L_1 + L_2$ are the coproduct morphisms, then $m$ defines the extension diagram:

$$
    \xymatrix{
          R_1 + L_2 \ar@{->}[d]_{} \ar@{}[rd]|{(1)} 
          & L_1 + L_2  \ar@{=>}[l] \ar@{=>}[r] \ar@{->}[d]^{m} \ar@{}[rd]|{(2)} 
          & L_1 + R_2 \ar@{->}[d]          \\
          H_1 
          & G  \ar@{=>}[l] \ar@{=>}[r] 
          & H_2  \\
  }
  $$
In particular, we have to prove that (1) and (2) define extension diagrams. Let us argue w.r.t. extension diagram (1) (we can argue analogously for (2)). 

We have to prove that (3) and (4) are pushouts:

\[\tikz[node distance=2em,shape=rectangle,outer sep=1pt,inner sep=2pt,label distance=-1.25em]{
\node(RL){$R_1 + L_2$};
\node(AGL)[node distance=3em,strictly right of=RL]{$I_1 + L_2$};
\node(AGR)[node distance=3em,strictly right of=AGL]{$L_1 + L_2$};
\node(BGL)[strictly below of=AGL]{$D_{1}$};
\node(HR)[strictly below of=RL]{$H_{1}$};
\node(BGR)[strictly below of=AGR]{$G$};
\draw[monomorphism] (AGL) --  (AGR);
\draw[monomorphism] (BGL) --  (BGR);
\draw[altmonomorphism] (AGL) --  (RL);
\draw[altmonomorphism] (BGL) --  (HR);
\draw[morphism] (RL) --  (HR);
\draw[morphism] (AGL) --  (BGL);
\draw[morphism] (AGR) -- node[overlay,right](m){\small $f$} (BGR);
\draw[draw=none] (RL) -- node[overlay](po1){\small (3)} (BGL);
\draw[draw=none] (AGL) -- node[overlay](po2){\small (4)} (BGR);
}\]

We have that squares (5), (6) and (7) below are pushouts, therefore (4) is also a pushout, according to the Butterfly Lemma (see  \cite{EhrigEPT06}). Similarly, since squares (8), (9) and (10) below are pushouts, for the same reason, (3) is also a pushout.

\[\tikz[node distance=2em,shape=rectangle,outer sep=1pt,inner sep=2pt,label distance=-1.25em]{
\node(GGL){$L_2$};
\node(GGK)[node distance=3em,strictly right of=GGL]{$D_{1}$};
\node(GGR)[node distance=3em,strictly right of=GGK]{$G$};
\node(GL)[strictly above of=GGL]{$L_2$};
\node(GK)[strictly above of=GGK]{$D_{1}$};
\node(GR)[strictly above of=GGR]{$G$};
\node(KK)[strictly above of=GK]{$I_{1}$};
\node(RR)[strictly above of=GR]{$L_{1}$};
\draw[draw=none] (KK) -- node[overlay](po1){\small (5)} (GR);
\draw[draw=none] (GL) -- node[overlay](po1){\small (6)} (GGK);
\draw[draw=none] (GK) -- node[overlay](po1){\small (7)} (GGR);
\draw[monomorphism] (KK) --  (RR);
\draw[morphism] (GL) -- node[above]{\small $d_{21}$} (GK);
\draw[monomorphism] (GK) -- node[above]{\small $k_1$} (GR);
\draw[morphism] (GGL) -- node[below]{\small $d_{21}$} (GGK);
\draw[monomorphism] (GGK) -- node[below]{\small $k_1$} (GGR);
\draw[morphism] (GL) --  (GGL);
\draw[morphism] (KK) --  (GK);
\draw[morphism] (GK) --  (GGK);
\draw[morphism] (RR) -- node[overlay,right](m){\small $m_1$} (GR);
\draw[morphism] (GR) --  (GGR);
\node(HGL)[node distance=6em,strictly right of=GGR]{$L_2$};
\node(HGK)[node distance=3em,strictly right of=HGL]{$D_{1}$};
\node(HGR)[node distance=3em,strictly right of=HGK]{$H_1$};
\node(HL)[strictly above of=HGL]{$L_2$};
\node(HK)[strictly above of=HGK]{$D_{1}$};
\node(HR)[strictly above of=HGR]{$H_1$};
\node(HKK)[strictly above of=HK]{$I_{1}$};
\node(HRR)[strictly above of=HR]{$R_{1}$};
\draw[draw=none] (HKK) -- node[overlay](po1){\small (8)} (HR);
\draw[draw=none] (HL) -- node[overlay](po1){\small (9)} (HGK);
\draw[draw=none] (HK) -- node[overlay](po1){\small (10)} (HGR);
\draw[monomorphism] (HKK) --  (HRR);
\draw[morphism] (HL) -- node[above]{\small $d_{21}$} (HK);
\draw[monomorphism] (HK) -- node[above]{\small $c_1$} (HR);
\draw[morphism] (HGL) -- node[below]{\small $d_{21}$} (HGK);
\draw[monomorphism] (HGK) -- node[below]{\small $c_1$} (HGR);
\draw[morphism] (HL) --  (HGL);
\draw[morphism] (HKK) --  (HK);
\draw[morphism] (HK) --  (HGK);
\draw[morphism] (HRR) -- (HR);
\draw[morphism] (HR) --  (HGR);
}\]

\end{proof}

{\noindent \bf Proof of \autoref{lem:initial-parallel-independent}}
\begin{proof}

By \autoref{lem:coproduct extensions}, we know that (1)+(2) is an extension diagram, where $m: L_1 + L_2 \rightarrow G$ is the mediating morphism for the coproduct $L_1 + L_2$.

$$
    \xymatrix{
          R_1 + L_2 \ar@{->}[d]_{} \ar@{}[rd]|{(1)} 
          & L_1 + L_2  \ar@{=>}[l] \ar@{=>}[r] \ar@{->}[d]^{m} \ar@{}[rd]|{(2)} 
          & L_1 + R_2 \ar@{->}[d]          \\
          H_1 
          & G  \ar@{=>}[l] \ar@{=>}[r] 
          & H_2  \\
  }
  $$

Let us now assume that $tp': H'_1\ldder_{p_1,m_1}G'\dder_{p_2,m_2} H'_2$ can be embedded in $tp$ via $f': G' \rightarrow G$, defining extension diagrams (5)+(6). 

$$
    \xymatrix{
          R_1 + L_2 \ar@{->}[d]_{} \ar@{}[rd]|{(3)} 
          & L_1 + L_2  \ar@{=>}[l] \ar@{=>}[r] \ar@{->}[d]^{m'} \ar@{}[rd]|{(4)} 
          & L_1 + R_2 \ar@{->}[d]          \\
          H'_1 \ar@{->}[d]_{} \ar@{}[rd]|{(5)} 
          & G'  \ar@{=>}[l] \ar@{=>}[r] \ar@{->}[d]^{f'} \ar@{}[rd]|{(6)} 
          & H'_2  \ar@{->}[d]         \\ 
                H_1  &G \ar@{=>}[r]\ar@{=>}[l] &H_2 
  }
  $$
We know that there is a unique morphism $m': L_1 + L_2 \rightarrow G'$, such that $g \circ i_1 = m'_1$ and $g \circ i_2 = m'_2$,  defining by \autoref{lem:coproduct extensions} the extension diagrams (3)+(4). Hence, we only have to prove that  
$f' \circ m' = m$, but we know that $m: L_1 + L_2 \rightarrow G$ is the unique morphism that defines the outer extension diagrams (3)+(4)+(5)+(6), thus $f' \circ m' = m$. \cqd
\end{proof}

{\noindent \bf Proof of \autoref{thm:infinite-set-complete}}
\begin{proof}
ACs over the empty graph $\emptyset$ in particular express so-called graph properties.   Graph properties formulated this way have the same expressive power as first-order logic (FOL) on graphs\footnote{FOL on graphs is standard first-order logic with two additional built-in predicates: $Node(n)$ -to state that $n$ is a node and $Edge(e,n,n')$ to state that $e$ is an edge from $n$ to $n'$.} as shown in~\cite{HabelP09}. This means that we can express any graph property
equivalently using a first-order formula. For the same reason, we can state any graph property for graphs without isolated nodes  
using a first-order formula (i.e., any graph property that, in particular, implies that the given graph has no isolated nodes).  

Now consider the following two rules $\prule_1 = \langle \emptyset \leftarrow \emptyset \rightarrow \emptyset, c \rangle$ and 
$\prule_2 = \langle\emptyset \leftarrow \emptyset \rightarrow 1_N, \true \rangle$, with $1_N$ the graph consisting of an isolated node and $c$ some property (expressible using ACs over the empty graph) about graphs without isolated nodes. That is, the first rule can be applied to a graph $G$, if $G \models c$, and it leaves $G$ unchanged; and the second rule, which is always applicable, adds an isolated node to $G$. Thus the set of transformation pairs associated to these rules is 
$\SC_0 = \{G \Longleftarrow_{\prule_1} G \Longrightarrow_{\prule_2} G\oplus 1_N \mid G\models c \}$. Note that all the transformation pairs in $\SC_0$ are AC conflicts, since $G\oplus 1_N$ does not satisfy $c$ having an isolated node, which means that the set of conflicts of $\prule_1$ and $\prule_2$ is precisely $\SC_0$. In particular, for any graph $G$  either $G\models c$\footnote{A graph property is an application condition over the empty graph $\emptyset$ (or, in the general case, the initial object in the category of graphical structures considered), thus composed of literals of the form $c = \exists(\emptyset \rightarrow G', c')$. In particular, we say that $G \models c$ if  $i_G\models c$ with $i_G$ the unique morphism from $\emptyset$ to $G$.}  and both rules can be applied to $G$ in a unique way (since there is a unique match $h: \emptyset \rightarrow G$), or $G\not \models c$ such that $\prule_1$ cannot be applied. This means that, for any $G$, there is at most one transformation pair  $G \Longleftarrow_{\prule_1} G \Longrightarrow_{\prule_2} G\oplus 1_N$ starting from $G$. Consequently, if $G$ and $G'$ satisfy $c$, any morphism $h: G \rightarrow G'$ defines an extension diagram between their associated transformations. For these reasons, even if it is an abuse of notation, given sets of transformation pairs $\SC_0$ ($ \SC$), we write $G \in \SC_0$ (resp. $\SC$) meaning $G \Longleftarrow_{\prule_1} G \Longrightarrow_{\prule_2} G\oplus 1_N\in \SC_0$ (resp. $\SC$).

Now let us assume that a finite set $\SC$ of conflicts for rules $\prule_1$ and $\prule_2$ exists that is complete w.r.t. parallel dependence.    
This means that $G \in \SC_0$ if and only if there is a 
$G' \in \SC$ and a morphism $h:G' \rightarrow G$. We know, by the property of epi-mono factorization, that any morphism $h:G' \rightarrow G$ can be decomposed into $h = m\circ e$ with $m$ mono and $e$ epi. Moreover, since $G'$ is assumed to be finite,  
there is a finite number of epimorphisms whose source is $G'$. Let $Epi_{G'}$ be the set $\{G'' \mid \text{there\ is\ an\ epimorphism\ } e: G' \rightarrow G'' \}$, then we would have that $G \in \SC_0$  if and only if there is a 
$G' \in \SC$, a $G'' \in Epi_{G'}$ and a monomorphism $m:G'' \rightarrow G$. 
Note that, by definition of satisfaction (cf.~\autoref{def:condition}), the property that there is a monomorphism $m:G'' \rightarrow G$ is equivalent to 
$G \models \exists(\emptyset \rightarrow G'',\true)$. Therefore $G \in \SC_0$  if and only if there is a 
$G' \in \SC$, and a $G'' \in Epi_{G'}$ such that $G \models \exists(\emptyset \rightarrow G'',\true)$. But this means that $G \in \SC_0$  if and only if there is a 
$G' \in \SC$  such that $G\models \big(\bigvee_{G'' \in Epi_{G'}} \exists(\emptyset \rightarrow G'', \true)\big)$, or equivalently $G\models c'$, where c' is the condition
$$ c' =\big(\bigvee_{\substack{G'' \in Epi_{G'} \\ G' \in \SC}} \exists(\emptyset \rightarrow G'', \true)\big).$$
This means however that $c$ and $c'$ are logically equivalent, but this is a contradiction, since it is not possible to represent any arbitrarily complex first-order formula in terms of a finite disjunction of existential atoms. Therefore, our assumption was wrong and $\SC$ is in general infinite.\qed
\end{proof}

{\noindent \bf Proof of \autoref{thm:ic-is-cp}}
\begin{proof}
Given some initial conflict $stp_K: \tuple{tp_K,\ac_K,\ac^*_K}$,  we have that its matches $\tuple{o_1,o_2}$ are in $\mathcal{E'}$, since either $tp_K$ is an initial conflict or the initial parallel independent transformation pair for the plain rules $p_1$ and $p_2$. Note that for plain initial conflicts it is shown in \cite{LambersBO0T18} that their matches belong to $\mathcal{E'}$, and the coproduct morphisms also belong to $\mathcal{E'}$ because of uniqueness of the $\mathcal{E'}-\mathcal{M}$ pair factorization and the coproduct property.   
We moreover have that $stp_K$ is in particular a conflict-inducing symbolic transformation pair.  This means that there exists a conflicting pair $tp_G : H_1\ldder_{\prule_1,m_1}G\dder_{\prule_2,m_2} H_2$ in its unfolding $\mathcal{U}(stp_K)$.  Consequently, from the corresponding extension diagram with the extension morphism $m:K\rightarrow G$ we can derive directly  that $m\circ o_i$ for $i=1,2$ satisfy the gluing conditions. Moreover, because of \autoref{lem:shift} and the fact that $m\circ o_i \models \ac_{L_i}$ for $i=1,2$ we know that $m\models \ac_K$. 

Finally, we have to show that $m \models \ac^*_K$. Assume that $tp_K$ is an initial conflict for the plain rules $p_1$ and $p_2$. In this case $\ac^*_K$ is always true such that $m \models \ac^*_K$. Assume that $tp_K$ equals the initial parallel independent transformation pair $tp_{L_1 + L_2}$ as in \autoref{fig:embedding-pi}. We know that $tp_G$ is conflicting. It cannot be a use-delete/delete-use conflict, since this would be a contradiction with $tp_{L_1 + L_2}$ being parallel independent for plain rules.  Thus $tp_G$ is an AC conflict. This means that either $\ac_{L_1}$ or $\ac_{L_2}$ are not satisfied by the extended matches into $H_2$ and $H_1$, respectively. Then it follows because of \autoref{lem:shift}, \autoref{lem:left},  and the fact that diagonal morphisms for plain parallel independence are unique w.r.t. making the corresponding triangles commute that $m \models \ac^*_K = \ac^*_{L_1 + L_2}$ with $\ac^*_{L_1 + L_2,d_{12}}=\Left(p^*_2,\Shift(c_2{\circ}d_{12},\ac_{L_1}))$ and $\ac^*_{L_1 + L_2,d_{21}}=\Left(p^*_1,\Shift(c_1{\circ}d_{21},\ac_{L_2}))$. 
\qed
\end{proof}

{\noindent \bf Proof of \autoref{thm:completeness-ic-acs}}
\begin{proof}
Given a parallel dependent pair of transformations $tp_G: H_1\ldder_{\prule_1,m_1}G\dder_{\prule_2,m_2} H_2$ we need to show that some initial conflict via rules $\prule_1$ and $\prule_2$ exists that can be embedded into $tp_G$ via some extension morphism $m:K\rightarrow G$ with $m \models \ac_K \wedge \ac^*_K$. 

Assume that $tp_G$ is a use-delete/delete-use conflict.  Then $tp_G$ is also a use-delete/delete-use conflict as AC-disregarding transformation pair.  This means that an initial conflict $tp_K$ for the plain rules $p_1$ and $p_2$ exists according to \autoref{thm:completeness-core} that can be embedded via some extension morphism $m:K\rightarrow G$ into $tp_G$ as AC-disregarding transformation pair.   The symbolic transformation pair $stp_K:\tuple{tp_K,\ac_K,\ac^*_K}$ is obviously conflict-inducing. We moreover show that $m \models \ac_K \wedge \ac^*_K$.  It follows that $m \models \ac_K$ because of \autoref{lem:shift} and the fact that the matches of $tp_G$ satisfy $\ac_{L_1}$ and $\ac_{L_2}$. Moreover $m \models \ac^*_K$ since $tp_K$ is an initial conflict (i.e. delete-use) for the plain rules $p_1$ and $p_2$ such that $\ac^*_K$ is true.

Assume that $tp_G$ is not a use-delete/delete-use conflict, but it is an AC conflict.  Since $tp_G$ is not a use-delete/delete-use conflict we know that it is parallel independent as AC-disregarding transformation pair.  This means that the initial parallel independent transformation pair $tp_{L_1 +L_2}: R_1 + L_2 \ldder_{p_1, i_1} L_1 + L_2 \dder_{p_2, i_2} L_1 + R_2$ for the plain rules $p_1$ and $p_2$ can be embedded via extension morphism $m: L_1 + L_2 \rightarrow G$ into $tp_G$ as AC-disregarding transformation pair (as illustrated in \autoref{fig:embedding-pi}). The symbolic transformation pair $stp_{L_1 +L_2}:\tuple{tp_{L_1 +L_2},\ac_{L_1 +L_2},\ac^*_{L_1 +L_2}}$ is obviously conflict-inducing. We moreover show that $m \models \ac_{L_1 +L_2}\wedge \ac^*_{L_1 +L_2}$.  It follows that $m \models \ac_{L_1 +L_2}$ because of \autoref{lem:shift} and the fact that the matches of $tp_G$ satisfy $\ac_{L_1}$ and $\ac_{L_2}$. Moreover $m \models \ac^*_{L_1 +L_2}$ with $\ac^*_{L_1 + L_2,d_{12}}=\Left(p^*_2,\Shift(c_2{\circ}d_{12},\ac_{L_1}))$ and $\ac^*_{L_1 + L_2,d_{21}}=\Left(p^*_1,\Shift(c_1{\circ}d_{21},\ac_{L_2}))$
 because of \autoref{lem:shift}, \autoref{lem:left}, the fact that diagonal morphisms for plain parallel independence are unique w.r.t. making the corresponding triangles commute, and the fact that $tp_G$ is AC conflicting (i.e. either $\ac_{L_1}$ or $\ac_{L_2}$ are not satisfied by the extended matches into $H_2$ and $H_1$, respectively). 
\qed
\end{proof}

{\noindent \bf Proof of \autoref{cor:minimal-AC}}
\begin{proof}
Assume that there exists such a set $\mathcal{S'}$. Let $\mathcal{S'}^{pl}$ be the equally sized set of plain transformation pairs via the rules $\tuple{p_1,p_2}$ derived from $\mathcal{S'}$ by extracting merely the corresponding plain transformation pairs. Let $\mathcal{K}$ be the set of initial conflicts for the plain rules $\tuple{p_1,p_2}$ that did not lead to an use-delete/delete-use initial conflict for the rules $\tuple{\prule_1,\prule_2}$, since their corresponding symbolic transformation pair is not conflict-inducing (because its unfolding is empty). We start with showing that the set $\mathcal{S'}^{pl} \cup \mathcal{K}$ is complete w.r.t. parallel dependence  for the plain rules $\tuple{p_1,p_2}$. Given some conflict $tp$ via $\tuple{p_1,p_2}$, then its initial conflict either belongs to $\mathcal{K}$, or not.  Assume that it does not belong to $K$. We know by \autoref{thm:completeness-core} that some initial conflict via $\tuple{p_1,p_2}$ can be embedded into $tp$ that leads to some symbolic transformation pair $stp$ in $\mathcal{S}$. We know by \autoref{def:initial-conflict} that then there exists a conflicting transformation pair via $\tuple{\prule_1,\prule_2}$.  Since $\mathcal{S'}$ is complete w.r.t. parallel dependence for $\tuple{\prule_1,\prule_2}$, there needs to exist indeed a transformation pair in $\mathcal{S'}^{pl}$ that can be embedded into $tp$.

Now let $\mathcal{P}$ be the set of all use-delete/delete-use initial conflicts for $\tuple{\prule_1,\prule_2}$ in $\mathcal{S}$. We continue by assuming the following three cases: the size of $\mathcal{S'}$ is strictly smaller, equal or strictly larger than the size of $\mathcal{P}$.

In the first case, we assume that the size of $\mathcal{S'}$ is strictly smaller than the size of $\mathcal{P}$. We argue that the size of $\mathcal{S'}^{pl} \cup \mathcal{K}$ is then also strictly smaller than the size of the set of initial conflicts $\mathcal{I}$ for $\tuple{p_1,p_2}$. Note that $\mathcal{I} = (\mathcal{I}\setminus \mathcal{K}) \cup \mathcal{K}$. Now $(\mathcal{I}\setminus \mathcal{K})$ consists of all initial conflicts for $\tuple{p_1,p_2}$ that lead to an initial conflict for $\tuple{\prule_1,\prule_2}$. This means that the size of $(\mathcal{I}\setminus \mathcal{K})$ equals the size of $\mathcal{P}$. This contradicts with \autoref{cor:minimal}, since we have found a set $\mathcal{S'}^{pl} \cup \mathcal{K}$ that has smaller size than $\mathcal{I} = (\mathcal{I}\setminus \mathcal{K}) \cup \mathcal{K}$, but is still complete w.r.t. parallel dependence for $\tuple{p_1,p_2}$.

In the second case, we assume that the size of $\mathcal{S'}$ equals the size of $\mathcal{P}$. Since $\mathcal{S'}$ has size strictly smaller than the size of $\mathcal{S}$, it holds that $\mathcal{S} = \mathcal{P} \cup \{stp_{L_1 +L_2}\}$ (by definition $\mathcal{S}$ contains at least $\mathcal{P}$). This means that there exists at least one conflicting transformation pair $tp$ for $\tuple{\prule_1,\prule_2}$ such that the initial parallel independent transformation pair $tp_{L_1 +L_2}$ can be embedded into $tp$. Consequently, $tp$ is in particular not a use-delete/delete-use conflict, but an AC conflict. Since $\mathcal{S'}$ is complete w.r.t. parallel dependence for  $\tuple{\prule_1,\prule_2}$ it needs to be possible to embed a symbolic transformation pair $stp_K$ from $\mathcal{S'}$ into $tp$.  Note that the plain transformation pair $tp_K$ underlying $stp_K$ must be parallel independent, since $tp$ as plain transformation pair is parallel independent. This is because it is possible embed the plain parallel independent transformation pair $tp_{L_1 + L_2}$ into $tp$ that then by initiality can be embedded also into $tp_K$. From there it is easy to argue that then also $tp_K$ must be parallel independent. 
Now $\mathcal{S}'^{pl}$ necessarily contains this parallel independent plain pair of transformations $tp_K$. Consequently $\mathcal{S'}^{pl}$ minus this transformation pair $tp_K$ and together with $\mathcal{K}$ would still be complete for the plain rules w.r.t. parallel dependence. The size of $\mathcal{S'}^{pl} \setminus \{tp_K\} \cup \mathcal{K}$ is strictly smaller than the size of $(I \setminus \mathcal{K}) \cup \mathcal{K}$. Consequently, we have again found a contradiction with \autoref{cor:minimal}. 

In the third case, we assume that the size of $\mathcal{S'}$ is strictly larger than the size of $\mathcal{P}$. This contradicts our assumption, since the size of $\mathcal{S'}$ cannot be strictly smaller than the size of $\mathcal{S}$, which at least contains $\mathcal{P}$.  
\end{proof}

{\noindent \bf Proof of \autoref{thm:complete-unfolding}}
\begin{proof}
Because each disjunctive unfolding of a regular initial conflict consists of a finite number of elements (see finite index set \autoref{def:condition}), the set $\cup_{stp \in \mathcal{S}}\mathcal{U^D}(stp)$ is finite as soon as the set $\mathcal{S}$ of all initial conflicts is finite.  

We now show that the set $\cup_{stp \in \mathcal{S}}\mathcal{U^D}(stp)$ (consisting of concrete transformation pairs) is also $\mathcal{M}$-complete w.r.t. parallel dependence. 
From \autoref{thm:completeness-ic-acs} we know that the set of initial conflicts $S$ (consisting of symbolic transformation pairs) is complete w.r.t. parallel dependence.  This means that there exists some $stp_K:\tuple{tp_K,\ac_K,\ac^*_K}$ with AC-disregarding transformation pair $tp_K: P_1\ldder_{\prule_1,o_1}K\dder_{\prule_2,o_2} P_2$  from $S$ that can be embedded into $tp_G$ via some extension morphism $m:K\rightarrow G$ with $m \models \ac_K \wedge \ac^*_K$.

Consequently, since $m \models \ac_K \wedge \ac^*_K$ 
we know that because of having only regular initial conflicts $m \models \vee_{i \in I} \exists (a_i:K \rightarrow C_i, c_i)$. This means that $m \models  \exists (a_i:K \rightarrow C_i, c_i)$ for some $i\in I$ meaning that there exists some $q_i:C_i \rightarrow G  \in \mathcal{M}$ such that $q_i \models c_i$ and $q_i \circ a_i = m$. Because of the Restriction Theorem for plain rules~\cite{EhrigEHP06} and the fact that $q_i$ is in $\mathcal{M}$ we know that there exists a  pair of plain transformations via matches $a_i \circ o_1$ and $a_i \circ o_2$ that can be embedded into $tp_G$ via extension morphism $q_i$. 
Now we have to show that the matches $a_i \circ o_1$ and $a_i \circ o_2$ of this transformation pair $tp_{C_i}$ indeed satisfy the conditions $\ac_{L_1}$ and $\ac_{L_2}$, respectively.  Moreover, we argue that the transformation pair $tp_{C_i}$ is conflicting. To this extent, consider the identity morphism $id_{C_i}$ satisfying trivially $c_i$. Consequently, $a_i \models \exists (a_i:K \rightarrow C_i, c_i)$, and because of regularity it follows that $a_i \models \ac_K \wedge \ac^*_K$. By the Embedding Theorem~\cite{EhrigGHLO12,EhrigGHLO14} it then follows that we indeed obtain a pair of transformations with $a_i \circ o_1$ and $a_i \circ o_2$ satisfying the rule ACs  $\ac_{L_1}$ and $\ac_{L_2}$, since $a_i \models \ac_K$ making it AC-consistent for both AC-disregarding transformations in $tp_K$ indeed. Moreover, because of Lemma 6.2 (characterization of parallel dependency with ACs) in~\cite{EhrigGHLO12,EhrigGHLO14}  $tp_{C_i}$ is also parallel dependent, since $a_i \models \ac^*_{K}$.

Since pushouts and pushout complements are unique up to  isomorphism this pair of transformations $tp_{C_i}$ (built for the matches $a_i \circ o_1$ and $a_i \circ o_2$) is indeed equivalent to some transformation pair  from $\mathcal{U^D}(stp_K)$.  As a consequence we have indeed found an extension diagram embedding $tp_{C_i}: D_{1,i}\ldder_{\prule_1,a_i \circ o_1} C_i \dder_{\prule_2,a_i \circ o_2} D_{2,i}$ in $\mathcal{U^D}(stp_K)$ into $tp_G$ via $q_i$.  
\qed
\end{proof}

{\noindent \bf Proof of \autoref{thm:monotone-NACs}}
\begin{proof}
This follows directly from \autoref{def:symbolic-pair} and the constructions~\cite{EhrigGHLO14} related to \autoref{lem:left} and \autoref{lem:shift}. In particular, $\ac_K$ 
(arising from shifting each rule NAC over the match morphisms into $K$) consists of a conjunction of NACs again, and $\ac^*_{K}$ becomes true or consists of a (non-empty) disjunction of PACs. We obtain by shifting (using \autoref{lem:shift}) each NAC over each PAC morphism ($\exists id_K$ in the case $\ac^*_{K}$ becomes true) a condition that is equivalent to a disjunction of literals of the form $\exists(a_i:K\rightarrow C_i,\wedge_{q \in Q} \neg \exists n_q)$.  
\qed
\end{proof}

Let us recall the definition\footnote{We assume that the class $\mathcal{Q}=\mathcal{M}$, since we for simplicity do not distinguish between morphisms used to satisfy (or violate) a graph condition ($\mathcal{Q}$-morphisms) and $\mathcal{M}$-morphisms (as analogously assumed in the previous seminal work w.r.t. rules with ACs~\cite{EhrigGHLO12,EhrigGHLO14}.} of critical pairs for rules with NACs~\cite{LambersEO06}, before showing that initial conflicts for rules with ACs as defined in this paper represent a conservative extension  in the sense of \autoref{thm:conservative-NACs}. 

\begin{definition}[critical pair]\label{cat:def:cp}\index{critical pair}
A \emph{critical pair} is a pair of direct transformations $K \stackrel{p_1,m_1}{\Rightarrow} P_1$ with $NAC_{p_1}$ and $K \stackrel{p_2,m_2}{\Rightarrow} P_2$
with $NAC_{p_2}$ such that:
\begin{enumerate}
\item
\begin{enumerate}
\item $\nexists h_{12}:L_1 \rightarrow D_2:d_2 \circ h_{12} = m_1$ and $(m_1,m_2)$ in $ \mathcal{E}'$
 \newline (use-delete conflict) \newline or
\item there exists $h_{12}:L_1 \rightarrow D_2$ s.t. $d_2 \circ h_{12} = m_1$, but for one of the NACs $n_1:L_1 \rightarrow N_1$ of $p_1$
there exists a morphism $q_{12}: N_1 \rightarrow P_2 \in  \mathcal{M}$ s.t. $q_{12} \circ n_1 = e_2 \circ h_{12}$, and thus, $e_2\circ h_{12} \not \models NAC_{n_1}$, and $(q_{12},m'_2)$ in $ \mathcal{E}'$ (forbid-produce conflict)
\end{enumerate}
or
\item
\begin{enumerate}
\item $\nexists h_{21}:L_2 \rightarrow D_1:d_1 \circ h_{21} = m_2$ and $(m_1,m_2)$ in $ \mathcal{E}'$
\newline (delete-use conflict) \newline or
\item there exists $h_{21}:L_2 \rightarrow D_1$ s.t. $d_1 \circ h_{21} = m_2$, but for one of the NACs $n_2:L_2 \rightarrow N_2$ of $p_2$
there exists a morphism $q_{21}: N_2 \rightarrow P_1 \in  \mathcal{M}$
s.t. $q_{21} \circ n_2 = e_1 \circ h_{21}$, and thus, $e_1\circ h_{21} \not \models NAC_{n_2}$, and $(q_{21},m'_1)$ in $ \mathcal{E}'$ (produce-forbid conflict)
\end{enumerate}
\end{enumerate}
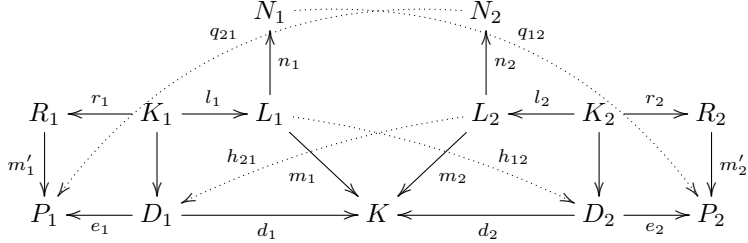
\[
\xymatrix{
 & & N_1 \ar@{.>}@/^2.3pc/[ddrrrr]^{q_{12}} & & N_2 \ar@{.>}@/_2.3pc/[ddllll]_{q_{21}} & & \\
R_1 \ar[d]_{m'_1} & K_1 \ar[r]^{l_1} \ar[l]_{r_1} \ar[d]& L_1  \ar@{.>}@/^/[drrr]^(.7){h_{12}} \ar[u]_{n_1} \ar[dr]_{m_1}& & L_2 \ar@{.>}@/_/[dlll]_(.7){h_{21}} \ar[u]_{n_2} \ar[dl]^{m_2}& K_2 \ar[d]\ar[l]_{l_2} \ar[r]^{r_2} & R_2 \ar[d]^{m'_2} \\
P_1 & D_1 \ar[rr]_{d_1} \ar[l]^{e_1}& &K & & D_2 \ar[ll]^{d_2} \ar[r]_{e_2} & P_2 \\
}
\]
\end{definition}

{\noindent \bf Proof of \autoref{thm:conservative-NACs}}
\begin{proof}
Recall that an initial conflict $stp_K:\tuple{tp_K,\ac_K,\ac^*_K}$ consists in particular of an initial conflict for plain rules or of an initial parallel independent pair of transformations for plain rules.  Having rules with NACs only, we can unfold such an initial conflict into a set of conflicting transformations $tp_{C_i}$ with each conflict stemming from one literal in the finite disjunction $\vee_{i\in I} \exists(a_i:K\rightarrow C_i,c_i)$ with $c_i$ a condition of the form $\wedge_{q \in Q} \neg \exists n_q$.  
When extending the initial parallel independent pair via some $a_i: L_1 + L_2 \rightarrow C_i$, the corresponding transformation pair remains plain parallel independent such that we in particular obtain a critical pair satisfying (1.b) or (2.b) according to \autoref{cat:def:cp}.  
Moreover we know that $(q_{12},m'_2)$ and $(q_{21},m'_1)$ belong to $\mathcal{E'}$ by construction and we know that $tp_{C_i}$ is AC-conflicting indeed.   
In case that we extend an initial conflict for the plain rules to a real conflict for the rules with NACs, we obtain a critical pair either satisfying (1.a) or  (2.a) according to \autoref{cat:def:cp} in case no additional context is added by the positive application condition $a_i$ stemming from the disjunction $\vee_{i\in I} \exists(a_i:K\rightarrow C_i,c_i)$ in the unfolding, or satisfying (1.b) or (2.b) as in the previous case.
\qed
\end{proof}


\end{document}